\documentclass[12pt,a4paper,fleqn,final]{article}
\usepackage[latin1]{inputenc}
\usepackage[english]{babel}
\usepackage{amsmath}
\usepackage{amsfonts}
\usepackage{amssymb}
\usepackage{makeidx}
\usepackage{graphicx}
\usepackage{mathpazo}
\usepackage{mathabx}
\usepackage{anyfontsize}
\usepackage[left=1 in,right=1 in,top=1 in,bottom=1 in]{geometry}
\usepackage{color}
\renewcommand{\baselinestretch}{1.10}
\small
\normalsize
\newtheorem{theorem}{Theorem}

\newtheorem{definition}{Definition}

\newtheorem{lemma}{Lemma}

\newtheorem{remark}{Remark}

\newenvironment{proof}[1][Proof]{\textbf{#1:} }{\ \rule{0.5em}{0.5em}}
\renewcommand{\baselinestretch}{1.1}
\small\normalsize

\usepackage[authoryear,round]{natbib}

\usepackage{natbib}

\title{\fontsize{14.5pt}{19pt}\selectfont\vspace*{-0.1in} Echo Chambers: \\
		{\fontsize{14.5pt}{10pt}\selectfont Voter-to-Voter Communication and Political Competition\setcounter{footnote}{1}\footnote{This paper is based on a chapter of my PhD dissertation at the University of Warwick. I am grateful to Mirko Draca, Walter Ferrarese, Carlo Perroni, Herakles Polemarchakis, Simone Tonin. I am also grateful to seminar participants at the Erasmus Political Economy Workshop and at the  Workshop for Woman in Political Economy.}
			\vspace*{0.1in}}}

\author{{\fontsize{12.5pt}{10pt}\selectfont%
	Monica A. Giovanniello}$^{\hspace*{0.02in}}$\setcounter{footnote}{3}\footnote{University of the Balearic Islands, Departament d'Economia Aplicada, Palma de Mallorca, 07122, Spain; e-maiul: ma.giovanniello@uib.cat; tel: +34 971 173437}}

\date{}

\begin{document}

	\pagestyle{empty}
	
	
	\renewcommand{\thefootnote}{\fnsymbol{footnote}} \setcounter{footnote}{0}
	
	\renewcommand{\baselinestretch}{1.05}
	\small \normalsize
	
	\maketitle

		\vspace*{-0.1in}
		
	\renewcommand{\baselinestretch}{1.1}
	
	\begin{quote}
		\normalsize 
		
		\begin{center}
			{\normalsize {\small Abstract}}
		\end{center}
		
		\begin{quote}
			{\normalsize {\small \ \ \ \ 
		I study how strategic communication among voters shapes both political outcomes and parties' advertising strategies in a model of informative campaign advertising. Two main results are derived. First, echo chambers arise endogenously. Surprisingly, small ideological distance between voters is not sufficient to guarantee that a chamber is created, biases' direction 
		plays a crucial role.
		Second, when voters' network entails a significant waste of information, parties target their advertising only to the voters leaning toward their opponent's rather than to their own supporters. }}
			
			\bigskip
			\smallskip
			\smallskip
			\noindent
			{\normalsize {\small {\sc Keywords}}:\ \ {\small 
				Echo Chambers; Social Networks; Information Aggregation; Political Advertising. }\smallskip }
			
			\noindent
			{\normalsize {\small {\sc JEL CLASSIFICATION}}: D72, D83, M37}
			
		\end{quote}
		
		\medskip
		
	\end{quote}
	
	\setcounter{page}{0} 
	
	\thispagestyle{empty}

	\pagestyle{plain} 
	
	\renewcommand{\baselinestretch}{1.32} 
	\small \normalsize
	
	\renewcommand{\thefootnote}{\arabic{footnote}}
	\setcounter{footnote}{0}

	\section{Introduction}
    
    One of the main characteristics of the new information society is how it horizontalizes information: the users of information are not just passive recipients, as was the case for traditional media (radio, television, newspaper), but are now able to spread information themselves, and in this way shape public opinion. 
	By mediating the exposure to information, these social media technologies have opened up an unexpected role for voters -- they now not only cast their votes but also have a greatly enhanced role in spreading information that can subsequently shape public opinion.\footnote{Donald Trump's main channel of communication was social media. Brad Parscale, digital director of Trump's campaign, stated that Facebook helped generate \$250 million in on-line fund-raising: ``Facebook and Twitter were the reason we won this thing''. Another example is the growth of the \textit{Five Star} movement in Italy, which emerged from a blog by the comedian Beppe Grillo. The \textit{Five Star} movement organizes face-to-face interactions via websites such as \textit{MeetUp.com} and solicits supporter views 
	on its central blog. Similarly, the UK's \textit{English Defense League} (EDL) emerged directly from a \text{Facebook} group centred on local race-related political developments in Luton. See \cite{scho2016} for a summary of these and other similar political movements.} 
    
    Earlier studies focusing on the consumption of ``top-down'' news ({\em i.e.} print newspapers, TV, on-line news websites) found evidence of information segregation in on-line news consumption (\citealp{gentzkow_ideological_2011}).\footnote{Previous works focused on the media bias in traditional news formats, such as high circulation newspapers (\citealp{mullainathan_market_2005}; and \citealp{gentzkow_media_2006}). A notable empirical exploration of the determinants of newspaper media bias is \cite{gentzkow_what_2010}.} 
    
    Nowadays this paradigm has changed: voters are not simply exposed to information; instead they interact and exchange information with deliberately selected sources and receivers. \cite{bakshy_exposure_2015} find that individuals' sharing choices plays a stronger role than the Facebook content recommendation algorithm in limiting exposure to cross-cutting content, voters network reinforces patterns of segregation. 
    Thus, any political information that originates from traditional outlets is mediated by social media, which creates scope for strategic information sharing, hiding, or misrepresenting.\footnote{\emph{''We just put information into the bloodstream to the internet and then watch it grow, give it a little push every now and again over time to watch it take shape. And so this stuff infiltrates the on-line community and expands ...,''} Mark Turnbull, managing director of Cambridge Analytica and SCL Elections (The  Guardian, E. Graham-Harrison and C. Cadwalladr, 21/03/2018).}  
    
    This paper provides a first theoretical investigation of how social network interactions and network structure influence the diffusion of politically-relevant information, and thus political outcomes.\footnote{Several empirical studies show that interactions within intimate social networks increases the likelihood of political participation (\citealp{mcclurg_social_2003}; \citealp{plutzer_identity_1996}; \citealp{rainie_social_2012}), and study the effect of social media on voters' exposure to political information (\citealp{bakshy_exposure_2015}; \citealp{BW2017}; \citealp{DiFonzo2017}; \citealp{garrett09echo}).}
    I develop a two-level game that embeds: \emph{(i)} a model of political competition where parties compete in campaign advertising, and \emph{(ii)} a model of personal influence in a peer-to-peer information network where voters can strategically communicate with each other in order to affect the policy outcome. 
    I focus on political competition between two policy-motivated parties. Each party selects, given its candidate type (policy position), the level of informative advertising in order to persuade voters. Parties' advertisement is costly, and it is assumed to be truthfully informative.\footnote{Parties invest in campaign to create advertisement able to capture the attention of voters and with which voter can emphasize. While voters communication is 'cheap' as they can tweet or share messages whose content has been already created.}
    Voters have preferences over policies, but do not know ex-ante the ideological position of the parties' candidates. Voters can obtain information about parties' candidate types either by direct exposure to the parties' advertisements or by strategically communicating with other voters within their network. Thus, voters update their beliefs and cast their vote sincerely. The party that obtains a simple majority of votes wins the election and enacts the policy corresponding to its candidate ideology.
    
    Two results are put forward. First, echo chambers arise endogenously.
    When voters can costlessly and strategically communicate with their friends they chose to share valuable information only with like-minded peers and do not share any valuable information otherwise. As result, information travels only in ideologically-aligned groups, this is what I mean by echo chambers. 
    I show that echo chambers arise if and only if voters are leaning toward the same party ({\em i.e.} voters are biased in the same direction), and their ideological distance is not too big. In other words, I show that small ideological distance between voters is not sufficient to guarantee that a chamber is created, biases' direction plays a crucial role. Thus, voters do not share valuable information neither with ideologically close voters who are biased toward the opposite party, nor with voters biased toward their favoured party if the ideological distance is not small enough. 
    
    Second, I analyse the effect of echo chambers, via social network's structure, on parties advertising strategies. 
    Echo chambers, i.e. an enclosed system in which information is credibly shared, can act as a hurdle to the spread truthful information among voters and restrict access to sources of reliable information. If voters network is characterized by few links and/or the probability that like-minded voters interact is low the network acts as a filter, parties advertisement is not echoed among voters and valuable information is lost in the communication stage. When this is the case parties maximize their probability of winning by targeting those voters who, if uninformed, would vote for their opponent rather than targeting their own supporters (whom if uninformed would cast their vote for them). If, instead, the voters' network is characterized by many links and/or by an high probability that like-minded voters interact the echo chambers act as information diffusion device. Parties maximize their probability of winning by randomly advertising their candidates, in this way they reach voters biased toward both parties, who will echo the advertisement within their network. 
    
    Two recent contributions examine the implications of information segregation in networks. \cite{bloch2018rumors} study under which conditions rational agents have incentive to spread rumours in social networks. They consider a model in which bias and unbiased agents can spread possibly false information to affect the voting outcome. They find that, for high priors beliefs about the desired state of the world, the social network acts like a filter: agents block messages received by their biased neighbours, in this way, the circulating messages, even if incorrect, may convey useful information to take the right collective decision. \cite{jann2016echo}, analyse a cheap talk game with multiple senders and receivers, and find that players self-select themselves into like-minded and homogeneous groups. Interestingly they find that, when players have different preferences and information, players self-segregation is Pareto improving. 
    
    However, differently from the current paper, neither study the mechanism behind the rise of an echo chamber, as I do here. I consider a situation in which voters with similar and dissimilar preferences freely communicate with each other and given the underlying social network structure, information travels from many-to-many without any stigma. I show that echo chambers are a product of voters communication strategies rather than environments predetermined by exogenous factors. 
    To the best of my knowledge, this paper is the first theoretical contribution that formalizes how network structure affect endogenous echo chambers formation, and their effect on both the information diffusion and parties advertisement strategies.  
    
    This paper is also related to the theoretical literature on communication flows (\citealp{crawford_strategic_1982}). A well-established result in this literature is that truthful information transmission is possible whenever the sender and the receiver have similar preferences, {\em i.e.} the bias gap is not too large (\citealp{austen_smith90}). 
    In contrast, I show that, although it is necessary that information is transmitted among voters with similar preferences, this is not sufficient to create an echo chamber, voters must also be biased toward the same direction (same party).
    
    Another related contribution is \cite{galeottiMatt_personal_2011} who investigate the effects of social learning on political outcomes in a model of truthful interpersonal communication and informative campaign advertising. 
    They find that non-strategic communication networks with a richer structure lead to political polarization. In contrast with their model, I focus on the effect of strategic communication transmission among voters on the political outcome. I find that, when voters can falsify or hide information to shape public opinion and affect the outcome of the election, information diffusion and aggregation depend on the structure of the communication network which, in turn, affect the size of echo chambers that can be exploited by parties. This different information environment highlights the importance of the voters' network structure and its effect on the political outcome.
    
    The remainder of the paper is organized as follows. Section 2 describes the model. Section 3 develops the benchmark model in which the only source of voters information is the parties' advertisement. 
    Section 4 characterizes the equilibrium of the communication game. Section 5 provides the characterization of the political equilibrium. In Section 6 the main assumptions of the paper are discussed. The last section concludes. All proofs are in Appendix.

    \section{The Model}

    	\noindent \textit{Parties.} 
    The policy space is one-dimensional. There are two policy-motivated parties, $ J \in \{ L, R \} $, that compete in an election. Each party is represented by a candidate who can be either a moderate ($m$) or an extremist ($e$).
    A candidate type, $t_J \in T = \{ e, m\} $, where $ 0 < m < 1/2$ and $e = m/2 $, is independently and exogenously assigned to each party.\footnote{The assumption $e = m/2 $ is made only to simplify computation exposition, as letting $e < m $ does not change the quality of my results.}$^,$\footnote{I relax this assumption later in Section 6.1 allowing the parties to strategically select also the candidate type.} With probability $\sigma_J$ the candidate of party $J$ is a moderate and with probability $1- \sigma_J$ he is an extremist. The probability distributions $ \sigma_J : T \rightarrow  [0,1] $, for $ J \in \{ L, R \} $, are common knowledge. A state of the world is, then, defined by the profile of candidate types for the two parties, $ \theta = (t_{L}, t_{R}) \in T \times T $.
    
    Abusing of notation, I denote by $t_L $ the ideology of the candidate of party $L$ and by $ 1 - t_R$ the ideology of party $R$'s candidate. 
    Parties have distance preferences over the policy space and their ideological bliss point coincide with the ideological bliss point of their extremist candidate, \emph{i.e.} party $L$ bliss point is $i_L = e$ and $R$'s bliss point is $i_R = 1- e$. 
    
	Parties privately observe their candidate type and simultaneously choose a level of informative campaign advertising, $ x_{J}(t_J) \in [0,1] $.\footnote{The level of informative advertisement, $ x_{J}(t_J)$, can be seen as the ability of the party's message to capture the attention of voters. The advertisements have the following two characteristics: parties can only advertise their own candidate, which implies there is no negative advertising; parties cannot lie about the ideological position of their candidates.} The advertising is truthful, but costly. 
	If party $J$ advertises its candidate with intensity $ x_{J}(t_{J}) $, each voter obtains perfect information about party $J$'s candidate type with probability $ x_{J}(t_{J}) $. 
	The party that wins the electoral competition enacts the policy preferred by its candidate. 
	
	Then, party $L$ chooses to advertise its candidate $ t_{L} $ with intensity $ x_{L}(t_{L}):  T \rightarrow [0, 1] $ to maximize its expected payoff:
	\begin{eqnarray}\label{Uparty}
    U_{L}(x \mid t_{L})&=& \sigma_{R}(m) [\pi_{L} ( x \mid \theta ) ( i_L - t_L )+(1- \pi_{L}(x))(i_L -(1-m))] \notag \\
    &+& \sigma_{R}(e) [\pi_{L} (x \mid \theta ) (i_L - t_L)+(1- \pi_{L}(x))  (i_L -(1-e))] - c x_{L}(t_L) \notag 
    \end{eqnarray}
	where $\sigma_{R}(t_R)$ is the probability that party $R$ has a candidate of type $ t_R $, $ \pi_{L} ( x \mid \theta ) $ is the probability that party $ L $ wins the election, given the state of the world $\theta$, with $ x = [ x_{L}(t_{L}), x_{R}(t_{R})] $ being the advertising levels chosen by the parties, 
	and $ c > 0$ is the advertising marginal cost.
	Substituting $i_L$ into (\ref{Uparty}), after some manipulation the  expected payoff of party $L$ can be expressed in a more parsimonious way:
	\begin{equation}\label{UpartL}
    U_{L}(x \mid t_{L}) = \sum_{ t_{R} \in \lbrace e , m \rbrace } \sigma_{R}(t_{R})[\pi_{L} ( x \mid \theta ) (1-t_{R}-t_{L})+ (e - (1- t_{R}))]- c x_{L}(t_{L}),
    \end{equation}
    and party $R$'s expected payoff is
    \begin{equation}\label{UpartR}
    U_{R}(x \mid t_{R}) =  \sum_{ t_{L} \in \lbrace e , m \rbrace } \sigma_{L}(t_{L})[\pi_{R} ( x \mid \theta  ) (1-t_{R} -t_{L}) + ( t_{L} - (1 - e))]- c x_{R}(t_{R}),
    \end{equation}	
    where $1- e$ is the ideological bliss point of party $R$.
    Notice that the utilities of the two parties are increasing in the ideological distance between their candidates, as $-\pi_J(t_J \mid \theta)(t_L - (1-t_R))>0$, and decreasing in the ideological distance between the opponent candidate and the party's bliss point, as $e - (1- t_{R})<0$ (resp. $t_{L} - (1 - e)<0$) for party $L$ ($R$). This implies that a party always prefer to lose against a moderate candidate and win with an extremist one.
     
    \vspace{0.3 cm}
    
    \noindent \textit{Voters.} There is a continuum of voters of unit measure which have distance preferences over candidates' positions, 
    \begin{equation}\label{Uvoter}
    u_{i}(t)= 
    \begin{cases}
    \quad - \mid i - t_{L} \mid $ \quad \quad \quad if  $  L  $ wins the election, $
    \\
    - \mid i - ( 1 - t_{R} )  \mid $  \quad\,\ if  $  R  $ wins the election, $
    \end{cases}
    \end{equation}
 
    where $ i \in (0,1) $ denotes a voter's ideological bliss point. There are two types of voters, partisans and independent. 
    To distinguish among different types of voter I introduce the following definition: 
    \begin{definition}\label{D1}
    Denote by $j= \{R, L \}$ the favourite party of voter $i$. The preferences of a partisan voter over the candidates type of the two parties are $ e_j \succ_i  m_j \succ_i m_{j^{-}}  \succ_i e_{j^{-}} $ with $ j^{-} \neq j $, i.e. she always votes for her favourite party regardless the candidates type of the two parties.
   The preferences of an independent voter over the candidate types of the two parties are $ m_j \succ_i m_{j^{-}} \succ_i  e_j \succ_i e_{j^{-}}$, i.e. she  always prefer to vote for the moderate candidate of the party she likes less over an extremist candidate of her favourite party.
    \end{definition}
    Using this voters' type definition and their preferences it is easy to show that ex-ante (before the advertising takes place) the ideological bliss points of the partisans' group are distributed on the intervals $ ( 0, 1/2 - m/4 )$ and $ ( 1/2 + m/4 , 1) $ for the left and right partisans respectively. While independents' voters bliss point are distributed on the interval $ [ 1/2 - m/4, 1/2 + m / 4 ] $. 
    To avoid trivial results I assume the two partisans groups have the same size and that the identity of the median independent voter is ex ante uncertain. The independent voters' ideological bliss points are uniformly distributed in the interval $ [ \mu - \tau , \mu + \tau] $, where $ \mu $ is drawn from a uniform distribution with support $ [ 1/2 - m/4, 1/2 + m /4 ] $, $\tau >0$, and $m < 1/4 - \tau/2$.

\begin{figure}[h]
    \includegraphics[scale=0.38]{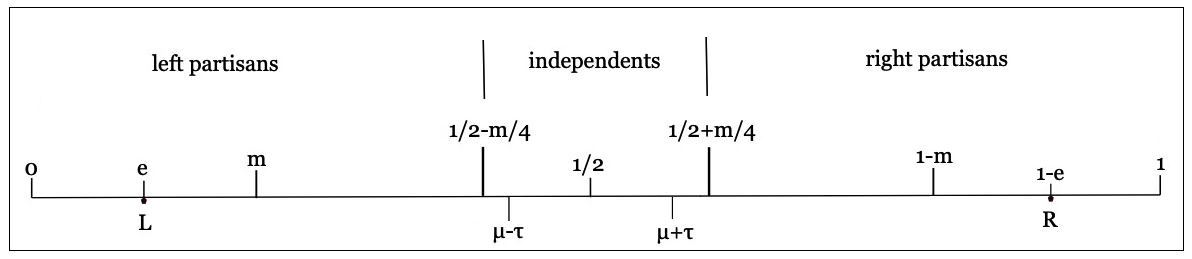}
    \caption{Parties, candidates, and voters ideologies}
\end{figure}

    \vspace{0.3 cm}

    \noindent \textit{Information Diffusion.} 
    Voters can obtain information about the candidates type either by direct exposure to the parties' advertisement or by strategically communicating with other voters. 
    
    Parties, after observing their candidate type, select the level of informative advertising, $x_J(t_{J})$. Then, each each voter is directly exposed to the information about $t_J$ with probability, $x_J(t_{J})$. 
    
    Each voter is exogenously matched with $z+k>0$ voters, where $k \geq 1$ is the number of her senders
    and $ z > 0 $ is the number of her receivers. 
    I assume that there is no preference uncertainty within the personal network of each voter, \emph{i.e.} each voter can perfectly observe the ideological bliss point of all her $z+k$ peers, but voters can not observe their peers' networks nor their voting decisions.\footnote{It is, in general, possible to distinguish two types of online networks: those based on common interests, e.g. Twitter, and those based on off-line relations, e.g. Facebook (Pew Research Center, Survey March 7 -- April 4 2016, ``Social Media Update''). In both type of networks voters have a good approximation of the ideological position of their peers. Although, as expected in the second type of networks the level of homophily within users' networks is reduced, this does not prevent information diffusion. Many empirical studies find networks based on intimate social relations are likely to stimulate interest in politics (\citealp{mac_carron_calling_2016}; \citealp{goncalves_modeling_2011}).}
    Thus, once the advertisement has taken place, all voters simultaneously receive $ k $ and send $z$ pairs of private messages about the candidate type of each party $( M_{L}(t_{L}), M_{R}(t_{R}) )$, 
    where $  M_{J}(t_{J}) \in \{e, m, \emptyset \} $ is a message about the candidate type $t_J$ of party $ J \in \{L , R\} $.\footnote{By assuming simultaneous communication among voters I take into account the life-span of new information on social media. As pointed out by \cite{hodas2012visibility}, on social media new pieces of information ''compete in time'', as the recommendation algorithms moves them to less visible queues, ''and space'', as users divide their limited attention among their peers.
    They show that in order to maximize the public users-voters tend to retweet information immediately after it has been tweeted for the first time, {\em i.e.} when it is most visible. In fact, given the limited attentions of the their peers users-voters visibility is crucial and so it is the timing of publication/retweet.} 
    The sender in each of the $k$ matched pair aims to maximize the probability that his favourite candidate wins the election. In other words, when sending the message pair to each of her receiver, the sender behaves as if she believes the receiver's vote would be the pivotal one. Then, the sender's intra-stage utility from each possible pair of messages is 
    \begin{eqnarray}\label{1sender}
    	u_{s}(M_{L}(t_{L}),M_{R}(t_{R}) \mid I^{s}, x, r ) &\hspace*{-0.2in}= &\hspace*{-0.1in}\Pr \left( v_{r} = 1 \mid I^{r}, x, r \right)  \mathbb{E} [ t_{L} - s \mid  I^{s}, x ]  \notag \\
    	&\ \ \ +&\hspace*{-0.1in} \Pr \left( v_{r} = 0 \mid I^{r}, x , r  \right)  \mathbb{E} [ s - (1 - t_{R}) \mid I^{s},  x], 
    \end{eqnarray}
    where $s$ denotes the bliss point of the sender in each matched pair, $ I^{s}$ is the information set of the sender and $ \Pr \left( v_{r} = 1  \mid I^{r},  x, r \right)  $ is the probability that the receiver with ideological bliss point $r$ votes for the party $L$, given her ex-post information set $ I^{r} $.

     \begin{figure}[h]
     \centering
     \includegraphics[scale=0.9]{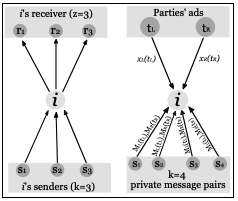}
     \caption{Left: Voter $i$ personal network. Right: Voter $i$ information sources}
     \end{figure}
        
    After the communication stage, the information set of voter $i$ becomes $ I^{i} =  $ $ (t_{L}, $ $ \overline{ M}_{L}(t_L),$  $ t_{R}, \overline{M}_{R}(t_R)) $, where $ t_{J} $ denotes the information about the candidate type of party $J$ acquired through direct exposure to the parties' advertisement, and $ \overline{M}_{J}(t_{J}) $ denotes the vector of messages received from the $k$ partners about the candidate type of party $J$. Let $ \rho(t_{L},t_{R} \mid I^{i}, x, k) $ denote the posterior belief of voter $ i $ about the state of the world, $ \theta = (t_{L},t_{R} ) $, where voters' prior beliefs about the true state of the world are consistent with $\sigma$. The posterior belief, whenever possible, is derived using Bayes' rule:
    \begin{eqnarray}\label{belief}
    \rho(t_{L},t_{R} \mid I^{i}, x, k)  & = & \frac{\Pr(t_{L},t_{R}) \Pr(I^{i} \mid t_{L},t_{R}, x, k)}{\sum_{t'_{L},t'_{R} \in T} \Pr(t'_{L},t'_{R}) \Pr(I^{i} \mid t'_{L},t'_{R}, x, k)  }.
    \end{eqnarray}	
   
    \vspace{0.3 cm}
    
    \noindent \textit{The indifferent independent voter.} After updating their beliefs, voters cast their vote sincerely. 
 	Voter $i$'s utility function is 
   \begin{small}
    \begin{equation}\label{Uvot}
 	    u_i(v_i \mid \theta, I, k) =\sum_{t_{L} \in T} \sum_{t_{L} \in T} \Big(  v_i \rho ( t_{L}, t_{R} \mid \theta, I, k ) (t_L - i) + (1-v_i)   \rho ( t_{L}, t_{R} \mid \theta, I, k) (i -(1-t_R)) \Big)
 	\end{equation}
    \end{small}
    where $v_i \in \{0,1\}$ is $i$'s voting decision and $v_i =1 $ (resp. $v_i =0 $) denotes the decision of voting for party $L$ ($R$).
    Note that voter's utility does not depend on her message strategy.\footnote{\label{gratific} Studies show that citizens share political news mostly motivated by the gratification of information seeking or for socializing rather than to political mobilize other voters (\cite{lee2012news}, \cite{kaye2004web}). 
    The communication intra-stage utility can be taken into account by adding to the voters utility function a gratification payoff $ \Phi w_i $, where $w_i \in\{0,1\}$ and $w_i = 1$ denotes the voter's decision of sending $z$ private message pairs that maximize the probability of obtaining his favoured outcome, and $\Phi>0$ is the resulting gratification. 
    } 
    Through strategic communication the partisan voters aim to maximize the probability that their party's candidate wins the election, while the independent voters aim to maximize the probability that a moderate candidate (preferably of their favoured party) wins the election. 
 	 
    Thus, given an information set $ I $, the voter votes for party $L$ if and only if 
 	\begin{equation}\label{votingdec}
 	u_i(v_i=1 \mid \theta, I, k) \geq u_i (v_i =0 \mid \theta, I, k). 
 	\end{equation}
    It follows that a voter with ideology $i$ and information set $ I $ votes for party $L$ if and only if $i < i^{*}(I)$, where $ i^{*}(I) $ solves the equation (\ref{votingdec}) with equality, \emph{i.e.} it is the bliss point of the indifferent independent voter with information set $I$. 
    The identity of the indifferent voter for any information set $ I $, is given by\footnote{See Appendix for derivation.}
    \begin{equation}\label{indifferent}
    i^{*}(I) =  \frac{1}{2} + \frac{1}{2} \displaystyle  \sum_{t_{L}, t_{R} \in T} \left(  \rho ( t_{L}, t_{R} \mid  \theta, I, k ) \, t_{L} - \rho ( t_{L},t_{R} \mid  \theta, I,  k ) \, t_{R} \right).
    \end{equation}
    Party $L$'s vote share is 
    \begin{displaymath}
    \mu^{*}( x \mid \theta, x, k) = \sum_{I} i^{*}(I)  \Pr( I \mid \theta, x,  k),
    \end{displaymath}
    then party $L$'s candidate wins the election with the following probability:
    \begin{equation}
    \pi_{L}( x \mid \theta, k ) = 
    \begin{cases}
    0 \hspace*{4 cm} \quad $ if $ \quad  \mu^{*}(x \mid \theta, k) < \frac{1}{2} - m
    \\
    \dfrac{\mu^{*}( x \mid \theta,  k) + m - \frac{1}{2}}{2m} \qquad \, $ if $ \quad   \mu^{*}(x \mid \theta, k ) \in \Big( \frac{1}{2} - m, \frac{1}{2} + m \Big)
    \\
    1  \hspace*{4 cm} \quad $ if $ \quad  \mu^{*}(x \mid \theta, k ) > \frac{1}{2} + m.
    \end{cases}
    \end{equation}

 	\vspace{0.3 cm}
 	
 \noindent \textit{Timing.}
    I study a symmetric five-stage Bayesian game. In the first stage parties, given their candidate type, select a level of informative campaign advertising, $x_J(t_J)$. In the second stage each voter observe the candidate type of each party $J$ with probability $x_J(t_J)$, and she is matched with $k+z>0$ other voters. In the third stage (communication stage) voters may be exposed to political information throughout strategic communication: each voter receives a pair of private messages, $(M_L(t_L), M_R(t_R))$, from her $k$ peers about the candidate type of the two parties.
    In the fourth stage, voters update their belief and sincerely cast their vote. Finally, the candidate that obtains a simple majority of votes wins the election and enacts the policy corresponding to his ideology.
	
\begin{figure}[h]
    \centering
    \includegraphics[scale=0.9]{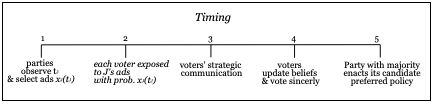}
    \caption{Timing of the game}
\end{figure}

    \medskip
    I analyse symmetric political equilibria, that is I assume the two parties have the same probability of competing in the election with a moderate candidate, $ \sigma_L = \sigma_R = \sigma$.

    \begin{definition}
    	\ A symmetric political equilibrium consists of: (i) symmetric parties' advertising strategies $ x^{*} = (x^{*}_{L}(t_{L}) ,x^{*}_{R}(t_{R}) ) $ ; (ii) a pair of message strategy for each sender in each matched pair $ M_{J}(t_{J}): T \rightarrow M $, for $ J \in \{ L, R \} $; (iii) voters' belief functions $ \rho(t_{L},t_{R} \mid I^{i}, x, k) $, (iv) indifferent independent voters type $ i^{*}(\cdot)$ such that:
    	\begin{itemize}
    		\item[1.] $ (x^{*}_{L}(t_{L}) ,x^{*}_{R}(t_{R}) ) $ are mutual best responses given the subsequent communication equilibrium and voting behaviour;
    		\item[2.]$  \rho^{*}(\cdot) $ is consistent and sequentially rational with the probability of the true state of the world $\sigma$, the optimal advertising strategies of the parties $ x^{*} $, and the senders' optimal message strategies $(M_L(t_L), M_R(t_R))$.
    	\end{itemize}
    \end{definition}

     \section{Benchmark Model}
    
    In this section I analyse parties advertisement strategies when 
    voters can not communicate with each other, that is $k =0$. 
    Thus, voters' information set is reduced to $I_0 = (t_L, t_R)$, which may contain only the information gathered from parties advertisement.\footnote{A seminal contribution on informative advertisement is \cite{grossman1984informative}. They analyse a model in which consumers choose among products of which they have seen the advertisement. As firms use advertising to create surplus they over-advertise, this way reducing profits. The main difference between the Grossman and Shapiro's model and my model is that voters may vote also for candidate they have not seen the advertisement. In other words, seeing and not seeing advertisement carries information as voters form expectations about them at the time of voting. This implies that Grossman and Shapiro's firms and the parties in this paper have different incentives. The firms over-advertise (even hurting their own profits) as it is the only way to create surplus, while the parties in this model obtains votes even when voters are uninformed.}
    
    Before proceeding with the analysis it is worth notice that
    a party that advertises an extremist incurs two costs: a direct cost of advertisement $c > 0 $, and an indirect cost due to disclosing harmful information. By advertising an extremist the party lowers the probability of winning the election as the independent voters, who are decisive for the outcome of the election, always prefer to vote for a moderate candidate rather than extremist.
    
    \begin{lemma}\label{lemcandidate}
    	Parties never advertise an extremist candidate.
    \end{lemma}
    \begin{proof}
    	See Appendix.
    \end{proof}

    Thus, voters' information about the candidate of each party $J$ gathered from advertisement can be summarized with $t_J = \{ m , \emptyset  \}  ,\ \forall J \in \{L , R\} $, and a voter is said to be perfectly informed about the candidate type of party $J$ whenever he observes advertisement $t_J = m $. An voter uninformed about the candidate type of party $J$, {\em i.e.} a voter who does not receive the advertisement, knows that with probability $\sigma (1 - x_{J}(m))/ 1- \sigma x_J(m)$ the $J$'s candidate is a moderate and with probability $ 1 - \sigma (1 - x_{J}(m))/ 1- \sigma x_J(m)$ he is an extremist.

    Working backwards, I first analyze the voting decisions. Once the advertisement has taken place the possible information sets hold by the voters are $\{(m,m); (m,\emptyset);$ $ (\emptyset,m); (\emptyset,\emptyset)\}$. 
   
   First, consider the case in which a voter holds an information set is $I_0^i= (\emptyset, m )$. The voter casts her votes for party $L$ whenever the equation (\ref{votingdec}) holds, i.e. whenever
    \begin{eqnarray}
     \rho ( m, m \mid I_0^i, x ) (m - i) + \rho( e, m \mid I_0^i, x ) \left( e - i \right) \geq i- (1-m), \notag
    \end{eqnarray}
   where $\rho ( e, m \mid I_0^i, x ) = 1 - \rho ( m, m \mid I_0^i, x )$ with $\rho ( m, m \mid I_0^i, x ) =  \frac{\sigma (1-x_L(m))}{ 1- \sigma  x_L(m) }$, and $e =m/2$. Which implies that a voter with information set $I_0^i= (\emptyset, m )$ votes for party $L$ if and only if her bliss point is 
    $$ i \leq \alpha^0_l \equiv \frac{1}{2} - \frac{m}{4} \left( 1 - \frac{\sigma (1-x_L(m))}{ 1- \sigma  x_L(m) } \right). $$
    Thus, a voter informed only about the candidate of party $R$ casts her votes for party $L$ if and only if her bliss point is $ i \in (0, \alpha^0_l ) $. By the symmetry of the game it follows that a voter informed only about the candidate of party $L$ ($I_0^i = (m, \emptyset)$) casts her votes for the party $R$ if and only if her bliss point is $ i \in ( \alpha^0_r , 1) $, where $\alpha^0_r  \equiv \frac{1}{2} + \frac{m}{4}  \rho ( m, e \mid I_0^i, x )$.
    By following the same steps above it is easy to show that, whenever a voter hold information set $I_0^i= (\emptyset,\emptyset)$ or $I_0^i= (m,m)$, she casts her vote for party $L$ if and only if she is $L$-leaning, \emph{i.e.} $i \leq 1/2$, and votes for party $R$ otherwise.
    
    Using Definition 1 it is easy to distinguish also the ex-post groups-types of voters, \emph{i.e.} whether, after the advertising, the voter with bliss point $i$ is a partisan or an independent. Once the parties advertise their candidates, if a voter's ideological bliss point is such that either $i \in (0, \alpha^0_l)$ or $i \in (\alpha^0_r  , 1)$ she is ex-post partisan, as her preferences are $ e_j \succ_i  m_j \succ_i m_{j^{-}}  \succ_i e_{j^{-}} $.
    If, instead, ex-post the voter's bliss point is $i\in ( \alpha^0_l ,\alpha^0_r)$ she is ex-post independent, as her preferences are $   m_j \succ_i m_{j^{-}}  \succ_i e_j  \succ_i e_{j^{-}} $.
    By symmetry, the interval defined by the two cut-off points always contains the centre of the ideological space (i.e. $\alpha^0_l < 1/2 < \alpha^0_r $), voters with ideological bliss point on the same side of the center point are said to be biased toward the same party.
    
       \begin{remark}\label{votypePostAdv}
    	There exist two ideological cut-off points, $\alpha^0_l$ and $\alpha^0_r$ which together with the central point define the ex-post group-type to which each voter belongs once the advertising has taken place:
    	\begin{itemize}
    	    \item if $i \in (0, \alpha^0_l) $ (resp. $i \in ( \alpha^0_r, 1 ) $) the voter is an ex-post $L$-leaning ($R$-leaning) partisan;
    	    \item if $i \in ( \alpha^0_l, 1/2) $ (resp. $i \in ( 1/2 , \alpha^0_r ) $) the voter is an ex-post $L$-leaning ($R$-leaning) independent.
    	\end{itemize}
    	
       \end{remark}
    
    From Lemma \ref{lemcandidate} follows that parties may advertise their candidate only if they run for the election with a moderate one. Given the symmetry of the game, it is sufficient to focus only to the case in which party $L$ has a moderate candidate. 
    
    In what follows I consider the case in which parties have access to two advertising technologies: i) a random technology, which allow each voter to be exposed at the advertisement with the same positive probability, ii) and group targeting technology which direct the advertisement exclusively to the voters who are leaning toward a given party (either only $L$-leaning or only $R$-leaning voters). 
    
    The benchmark model captures the minimum cost under which parties are willing to advertise their candidates when the two technologies are available, but voters can not communicate among each other.

    \paragraph{Random Advertising.} 
    First I analyse the case in which party $L$ can only randomly advertise its candidate. 
    
    Suppose that party $R$ has an extremist candidate. Whenever this is the case what matters is the distribution of indifferent voters, which in turn depends on the information set held. Given that party $L$ has a moderate candidate and it randomly advertises him with intensity $x_L(m)$, the proportion of uninformed voters is $1- x_L(m)$. While no voter receive any information from party $R$ as $x_R(e) = 0$.  
    Then party $L$ votes share, given any possible information set $I_0$, is
    \begin{eqnarray}
     \mu^0( x \mid \theta) & = & \sum_{I_{0}} i^{*}(I_0)  \Pr( I_0 \mid \theta, x)  = \frac{1}{2} (1- x_L(m)) + \left(  \frac{1}{2} + \frac{m}{4} \rho( m, e \mid I, x) \right)  x_L(m) \nonumber\\
     & = & \frac{1}{2} + \frac{m}{4} (1- \rho(m, m \mid I, x))  \, x_L(m), \nonumber \notag
    \end{eqnarray}
    which implies that party $L$ wins the election with probability
 	\begin{displaymath}
 	\pi^0(m,e \mid x) = \frac{1}{2} + \frac{1}{8} \,(1- \rho(m, m \mid I, x)) \, x_{L} .
 	\end{displaymath}
    If both parties have a moderate candidate and both randomly advertise them each party wins the election with probability $\pi(m,m \mid x) = 1/2$.
    Thus, party $L$'s expected utility from randomly advertising a moderate candidate, when $k=0$, is 
 	\begin{small}
 	\begin{eqnarray*}
 		U^0_{L}(x_{L} \mid \theta) & = &  \sigma \Big[\frac{1}{2}(1-2 m)- \frac{2-3m}{2} \Big] \\
 		& + & (1-\sigma) \Big[ 4 +  \frac{2-3m}{16} \rho(m, e \mid I^0, x)  x_{L}(m)  - (1-m) \Big] - c x_{L}(m).
 	\end{eqnarray*}
 	\end{small} 
 	which is linear in $x_L(m)$. Then, it follows that party $L$ advertises its candidate if and only if $ U^0_{L}(x_{L} \mid m, t_R) > U^0_{L}(x_{L}=0 \mid m, t_R)$, 
 	where 
 	\begin{small}
 	\begin{eqnarray*}
 	U^0_{L}(x_{L}=0 \mid m, t_R) &= & \sigma \Big[\pi^0_R(m,m \mid x_L=0, x_R) (1-2 m)- \frac{2-3m}{2} \Big] \\
 		& + & (1-\sigma) \Big[ \frac{1}{2} \left(1- \frac{3}{2} m \right) -(1-m) \Big] - c x_{L}(m)
	\end{eqnarray*}
 	\end{small} 
 \noindent	Thus, whenever the cost of advertisement is such that 
    \begin{equation}\label{0cost}
 	c \leq c^0 \equiv \frac{1}{16} (1- \sigma) (2- 3 m)(1- \rho(m,m \mid I^0, x))
    \end{equation}
 	party $L$ advertise with intensity $\bar{x}_{L}^{0} (m)= 1$, and never advertise a moderate candidate otherwise.

 	This implies that, if $c < c^0 $, a voter uniformed about the candidate of party $J$ would anticipate that the party is running with an extremist candidate. Thus, if both parties have a moderate or an extremist candidate, each wins the election with probability $\pi^0 (m,m \mid x) = \pi^0 (e, e \mid x)=1/2$, while if only one party has a moderate candidate it wins the election with probability $\pi^0 (m,e \mid x)= \pi^0 (e, m \mid x) = 1$.\footnote{This result depends on the specific functional form of the cost function. Suppose the cost function is quadratic, $c(x) = c \frac{x^2}{2}$. Then parties advertise a moderate candidate if and only if $c \leq c^{00} \equiv \frac{1}{8} \left( (2-3m)(1-\sigma) \rho (m, e \mid I^0, x) + 2 (1-2 m)(1-\rho (m, e \mid I^0, x)) \sigma \right)$ and $\bar{\bar{x}}^{0}_J(m) = \frac{(2-3m)(1-\sigma) \rho (m, e \mid I^0, x)}{16 c } $. From which follow that uninformed voters can not anticipate with certainty the candidate type when $c < c^{00}$.  }
 	
    \paragraph{Targeted Advertising.} Suppose parties can target their advertisement to only one ideological group of voters, { \em i.e. }  either to the $L$-leaning or to the $R$-leaning, and all voters that belong to the chosen group receive the advertisement with probability one, while all the voters that belong to the other group do not receive it. Thus, each voter that belongs to the targeted group $j$ will receive the advertisement with probability $x_{Jj}^{d}(t_J) = 1$, where $d $ denotes that party $J$ is using a targeting technology which directs the information exclusively to the $j$-leaning voters.
    
    In what follows I consider only the case in which $c > c^0$, as whenever $c \leq c^0$ party advertise their candidate to the entire population, $\bar{x}^0(m) =1$, which clearly dominates any targeting strategy.
    
    Whenever $c > c^0$ parties prefer to do not randomly advertise their candidate. To simplify the argument, for the moment suppose that only party $L$ can deviate by targeting its advertisement. If party $L$ targets $L$-leaning voters (its own supporters): all left leaning voters are informed about the candidate's type of their favoured party, and $L$ wins the election with probability $ 1/2 $, regardless of the candidate type of party $R$. By not advertising party $L$ obtains the same probability of winning the election but at a zero cost, it follows that party $L$ never targets its advertisement to its own supporters.
    
    If, instead party $L$ targets the voters leaning toward its opponent (henceforth opponent's supporters) it wins the electoral competition with probability $ \pi_{L}(m, x_R(t_R) \mid x^{\tau}  ) = 1 $, as I assumed that party $R$ has no access to targeting technology for now and $x_R(t_R) = 0$ for all $t_R$ when $c > c^0$.
    Then, party $L$ target its advertisement to the opponent supporters, $x_{JJ^{-}}^{\tau}(t_J) = 1$, whenever its utility from targeting its advertisement, $U_L^{\tau^0} (x_L^{\tau}(m), x_R(t_R) \mid \theta) $, is higher than the utility it receives when it does not advertise its candidate, $U^0_{L}(x_{L}(m)=0, x_R(t_R) \mid m, t_R)$. Thus, party $L$ targets its advertisement to the opponent supporters whenever 
    \begin{equation}
      c  \leq  \frac{(2-3m)(1-\sigma)}{4}.
    \end{equation}
    By symmetry it follows that also party $R$ has an incentive  target its advertisement to $L$-leaning voters. 
    This implies that, if a party runs with a moderate, has access to targeting technologies, and $ c  \leq  \frac{(2-3m)(1-\sigma)}{4}$, it wins the election with  with probability $1/2$ if both run with a moderate candidate, and with probability one if the opponent has an extremist. Then, it is easy to show that if 
     \begin{equation}
     c < c^{\tau} \equiv \frac{(2-3m - m \sigma)}{4},
     \end{equation}
    where $ c^{\tau} >  c^0$, parties advertise a moderate candidate by targeting their advertisement to voters leaning for their opponent. 
    
    	\begin{theorem}\label{benchRandom}
    In a symmetric pure strategy equilibrium without voters communication, parties advertise a moderate candidate with intensity $x^0(m)^{*}= 1$ if and only if the unit cost of advertisement $c$ is such that
 	\begin{equation*}
 		c \leq c^0 \equiv \frac{1}{16} (1- \sigma) (2- 3 m) (1-\rho(m,m \mid I^0, x)),
 		\end{equation*}
 	target their advertisement to the opponent supporters, $x_{JJ^{-}}^{\tau}(t_J) = 1$ , if and only if 
 	 \begin{equation*}
      c <  c^{\tau} \equiv \frac{(2-3m - m \sigma)}{4},
    \end{equation*}
    where $ c^0  < c < c^{\tau} $, and never advertise a moderate candidate otherwise.
 	\end{theorem}
    
    When voter's communication does not takes place independent voters may cast their vote for their less favoured party if and only if they are hit by its advertisement. Then, in equilibrium, if the cost of the advertisement is low, parties advertise with intensity $x^*(m)=1$ a moderate candidate. In other words, parties inform the entire population of their candidate type if the advertising technology is efficient enough. If, instead, the cost of advertisement is high parties prefer to target their advertisement to the opponents' supporter and let their own supporter uninformed, this way to guarantee the victory of their moderate whenever it runs against an extremist and win with probability 1/2 otherwise.\footnote{Policy motivated parties always prefer to lose against a moderate and win against an extremist.}

    \section{Echo Chambers}
    
    In this section I introduce the communication game assuming that each voter is linked to $k \geq 1$ senders and $z\geq 1 $ receivers. In the communication game, then,  information travels from many to many, given the underlying social network structure, and I show that echo chambers arise as a product of voters' communication strategies.  
    
    Once the advertisement has taken place, each voter receives $k \geq 1$ (resp. sends $z\geq 1$) private messages about the candidate's type of both parties ($ M_L (t_L), M_R (t_R) $) from each of her $k $ senders (resp. to her $z$ receivers). 
    
    Working backwards I start analysing the voters decision for each possible information set, in this way I can characterize the ex-post communication groups (henceforth simply ex-post groups) of partisan and independent voters. 
    
    In the last stage voter $i$'s decision is 	
    \begin{equation}
    v_{i} =
    \begin{cases}
    \ 1 $\quad\quad if $ \mathbb{E} \left[t_{L}- i \mid I^{i}, x, k \right] \geq \mathbb{E} \left[ i -(1-t_{R}) \mid I^{i}, x, k \right],
    \\
    \ 0 $\quad\quad otherwise, $
    \end{cases}
    \end{equation}
    where $x$ is the equilibrium advertising strategy vector of the two parties, and $ I^{i} = (t_L , M_L(t_L) , t_R , M_R(t_R) ) $ voter's $i$ information set. 
    
    For each voter's personal network $k$, subsequent communication strategies, and advertising strategies, voters can be divided into ex-post  partisans and ex-post  independents. 
    
    \begin{lemma}\label{alphak}
       	There exist two ideological cut-off points, 	$$\alpha_l =  \dfrac{1}{2} - \dfrac{m}{4} \left[ 1 - \rho (m,m \mid I, x, k)\right] \quad \text{and} \quad \alpha_r =  \dfrac{1}{2} + \dfrac{m}{4} \left[ 1 - \rho (m,m \mid I, x, k)\right], $$ 
       	which together with the central point define the ex-post group-type to which each voter belongs once both the advertising  has taken place and each voter has received the $k$ message pairs:
    	\begin{itemize}
    	    \item if $i \in (0, \alpha_l) $ (resp. $i \in ( \alpha_r, 1 ) $) the voter is an ex-post left-leaning (right-leaning) partisan;
    	    \item if $i \in ( \alpha_l, 1/2) $ (resp. $i \in ( 1/2 , \alpha_r ) $) the voter is an ex-post left-leaning (right-leaning) independent.
    	\end{itemize}
    
    \end{lemma}
    \begin{proof}
    	See Appendix.
    \end{proof} 
    
    Whenever voters' ideological bliss point $i$ is such that  $i \in (\alpha_l , \alpha_r)$, voters are ex-post independents; otherwise they are ex-post partisans (see Appendix). More specifically, if a voter's ideological bliss point $i$ is such that:
 $i \in (0 , \alpha_l)$ (\emph{resp.} $i \in (\alpha_r, 1)$ ) the voter is an ex-post $L$-leaning (\emph{resp.} $R$-leaning) partisan;
$i \in ( \alpha_l, \alpha_r)$ the voter is an ex-post independent voter, and she is said to be $L$-leaning (\emph{resp.} $R$-leaning) whenever $i \in (\alpha_l, 1/2) $ (\emph{resp.} $i \in (1/2, \alpha_r) $.
    

	Before proceeding with the analysis it is worth noticing that from Lemma \ref{lemcandidate} follows that none of the voters can credibly convey a message $M_J(e)$, given that voters have zero probability of receiving information $t_J = e$ from parties advertising. Then, it must that in equilibrium the sender is indifferent between $M_J(e)$ and $M_J(\emptyset)$ as the two messages can convey the same information. Thus, in what follows, without loss of generality, I assume that a sender indifferent between two pairs of message always prefers to send a truthful one, and this allows me to restrict the message's space to $ M_J (t_J) = \{  m ,  \emptyset \} ,\ \forall J \in \{L , R\}  $. 
    
    In the next sections I first derive the properties of a pairwise communication equilibrium (Section 3.1), and then I characterizes the echo chambers equilibrium (Section 3.2).

    \subsection{Pairwise Communication}
        
    To derive the main result of the paper it is useful to analyse the basic network structure to highlight the role played by the direction of the ideological biases in the communication game.
    I assume that each voter is matched with only one sender (\emph{i.e.}, $k=1$), and, abusing of notation, I denote by $ s $ and $ r $ the ideological bliss points of the sender and the receiver, respectively.
    
    In the last stage the receiver's voting decision depends on her information $I^r$, that is on the information she may receive from each party advertisement ($t_L,  t_R$) and/or from the sender $(M_L(t_L), M_R(t_R))$, the parties' advertising strategy $x = \left( x_L , x_R\right) $, and  the ideological bliss point of her sender $s$. Thus, the receiver votes for party $L$ if and only if 
    \begin{equation*}
	\mathbb{E} \left[t_{L}- r \mid  I^{r}, x, s \right] \geq \mathbb{E} \left[ r -(1-t_{R}) \mid  I^{r}, x, s \right].
    \end{equation*}
    
    In turn, this implies that the informational incentive of the sender depends not only on the information he holds, but also on both the receiver's information and receiver's ideological bliss point. Recall that the sender's intra-stage utility from each possible pair of messages is 
    \begin{eqnarray}\label{sender}
    	u_{s}(M_{L}(t_{L}),M_{R}(t_{R}) \mid I^{s}, x, r ) &\hspace*{-0.2in}= &\hspace*{-0.1in}\Pr \left( v_{r} = 1 \mid I^{r}, x, r \right)  \mathbb{E} [ t_{L} - s \mid  I^{s}, x ]  \notag \\
    	&\ \ \ +&\hspace*{-0.1in} \Pr \left( v_{r} = 0 \mid I^{r}, x , r  \right)  \mathbb{E} [ s - (1 - t_{R}) \mid I^{s},  x].
    \end{eqnarray}
    
    Assuming that an indifferent sender sends a truthful messages, truthful revelation is incentive-compatible for the sender if and only if
    \begin{equation}\label{IC}
    u_{s} (M_{L}(t_{L}),M_{R}(t_{R}) \mid \theta, I^{s}, x, r ) \ \geq\ u_{s} (M_{L}(\hat{t}_{L}),M_{R}(\hat{t}_{R})\mid \theta, I^{s},  x, r  ).
    \end{equation}
    
    I state the following result.
	\begin{lemma}\label{paircheaptalk}		
    	In the pairwise cheap-talk game truthful revelation is possible if and only if 
    	\begin{itemize}
    		\item[(i)] the receiver is an ex-post partisan;
    		\item[(ii)] the $J$-leaning sender knows his favoured party runs with a moderate, i.e. $t_J= m$, and is uninformed about the candidate type of the party she likes less ($t_{J^{-}} = \emptyset$);
    		\item[(iii)] both the sender and the receiver belong to the same ex-post group type (partisan/independent) and are on the same side of the ideological center point.
    	\end{itemize}
    \end{lemma}
    \begin{proof}
    	See Appendix.
    \end{proof} 
    
    The intuition behind point (i) of Lemma \ref{paircheaptalk} is trivial. 
    The sender always sends a truthful message to an ex-post partisan receiver (\emph{i.e.}, either $r < \alpha_l $ or  $r > \alpha_r $) as she always vote for her favoured party regardless the information held. 
    
    In what follows I briefly analyse the incentive of the sender in the communication game when matched with an ex-post independent receiver.
   
    First consider case in which the sender is partially informed about the candidate type of his favoured party, \emph{i.e.} he only knows that his favoured party has a moderate candidate and has no information about the candidate's type of the opponent. 
    In this case, by lying to an ex-post independent receiver the sender may harm the probability that his favoured party wins the election with a moderate candidate. By sending a truthful message the sender never harms the probability that his favoured party wins the election and he neither harms the probability that a moderate is elected, given that an ex-post independent voter always prefer to vote for a moderate. 
    Thus, whenever a sender is partially informed about the candidate type of his favoured party always sends a truthful message (Lemma \ref{paircheaptalk} (ii)).
    
    To shortly illustrate the core argument of point (iii) of Lemma \ref{paircheaptalk} suppose the sender is left leaning, and that he knows that the candidate of party $R$ is a moderate, but has no information about the candidate type of his favoured party, $L$.
   
   An ex-post partisan sender knows that by truthfully sharing his information with an ex-post independent receiver always increases the probability that party $R$ wins the election, while he would prefer the left party to win the election (even with an extremist). Thus, an ex-post partisan sender does not truthfully transmit his information to an ex-post independent receiver, even if his receiver is biased toward the party the sender prefers.\footnote{It is easy to show that the same reasoning applies to the case in which the ex-post partisan sender is uninformed. If perfectly informed about both candidates' type the ex-post partisan sender truthfully share his information with voters leaning toward the party he prefers, and lies otherwise.}
   
   If the sender is an ex-post independent there are two possible scenarios: 1) sender and receiver are biased toward the same party (\emph{i.e.} the ideological bliss points of both voters are on the same side of the central ideology); 2) the two voters are biased toward different parties. 
   In the first case, the sender knows that by truthfully revealing her information he does not harm the probability of his favourite moderate to being elected. In fact, a $L$-leaning receiver votes for party $R$ if and only if he knows the $R$'s candidate is a moderate and has no information about the $L$'s candidate type (like the sender would have done). Thus, when the two voters are ex-post independent and their biases are aligned the sender, by truthfully sharing his information, increases the probability of a moderate being elected without harming the probability of election of a moderate of her favoured party $L$. 
   
   In the second case truthful revelation by the sender harms the probability that his favourite party with a moderate candidate wins the election: the $R$-leaning receiver uses the sender's information to vote for the sender's opponent party, regardless of the information he holds about the sender's favourite party's candidate ($L$).
     
   Thus, while an ex-post partisan receiver always obtains truthful information in the pairwise communication game, an ex-post independent receiver obtains truthful information in two cases: 1) both the sender and the receiver have aligned preferences (they are biased toward the same party) and belong to the same ex-post group type (partisans/independents); 2) the sender is informed about the party's candidate type he prefers and has no information about the candidate type of the opponent.

    \begin{lemma}\label{lemSeq}
    	Whenever voters in the matched pair have divergent preferences and/or do not belong to the same group type, the sender's truthful messages are not credible in equilibrium.
    \end{lemma}
    

    \subsection{Social Networks and Echo Chambers}
    
    In this section I characterize communication equilibrium in the social network. The structure of this equilibrium takes a specific form, which I refer to as {\em echo chambers}.
    
    Each voter personal network is characterized by $k \geq 1$ senders, by $ z > 0 $ receivers, and by the probability, $\beta \in (0, 1)$, that both $i$ and her $z+k$ peers are leaning toward the same party (degree of homophily). 
    Voters' networks are acyclic and the richness of the network, {\em i.e.} the number of links of each voter ($k$ and $z$), and the degree of homophily, $\beta $, are common knowledge. 
    All voters simultaneously receive $ k $ and send $z$ pairs of private messages about the candidate type of each party $( M_{L}(t_{L}), M_{R}(t_{R}) )$ from and to their peers.\footnote{Allowing the information to travel for more than one step complicate the analysis, but it does not affect the qualitative features of my results. In each matched pair of voters an ex-post independent receiver discards the message pairs received from both ex-post partisan senders and from ex-post independent senders if they are leaning for her less favoured party. Thus, the receiver behaves as if she never received such messages. In instead the voters receives credible message pairs then she updates her believes and truthfully shares his information with her next-step receiver if and only if either they are both ex-post independent with aligned preferences or the receiver is an ex-post partisan. This implies that truthful and credible information travels along paths in which the Lemma \ref{paircheaptalk} and \ref{lemSeq} hold in each node and it is blocked otherwise.}
    
    Now I am ready to state my main result:
    \begin{theorem}\label{Echo}
    	\ In any pairwise $k$-player cheap-talk game there are two ideological cut-off points, $ 0 < q_{l}(k,\beta, x_L) <1/2< q_{r}(k,\beta, x_R) < 1 $, such that truthful revelation in each matched pair is possible and credible if and only if the bliss points of both the sender and the receiver belong to the interval defined by those cut-off points and their preferences are aligned (\emph{i.e.} both are biased for the same party).
    \end{theorem}
    \begin{proof}
    	See Appendix.
    \end{proof} 
    
    As voter $i$ has $k+1$ sources of information, the parties' advertisement and the $k$ message pairs he receives from their network's peers, the sender's informational incentives depends not only on the information he gathers trough exposure to the parties' advertisement and the receiver's ideological bliss point, but also on the information that each of his receivers gather from the other $k-1$ senders.
    
    Given that within the sampled network of each voter there is no preference uncertainty and the communication among voters has the form of private messages, each voter evaluates the $k$ messages pairs received one by one. It follows that if a voter receives at least one credible and informative message about the candidate's type of party $J$ then he knows that the $J$'s candidate is a moderate with probability one.
    Conversely, if the receiver observes $k$ credible but uninformative messages about party $J$ candidate type then he knows that the candidate of party $J$ is a moderate with probability $ \sigma (1-x_{J}(m))^{k} / (1-x_{J}(m))^{k} \sigma + 1  -\sigma $ and an extremist otherwise. 
    Given the voters' personal network structure there exist two ideological cutoffs,
    \begin{displaymath}
    q_{l}(k,\beta, x_L) = \frac{1}{2} - \frac{m}{4} \frac{ 1-\sigma}{1-\sigma + \sigma (1-x^{*}_{L}(m))^{\beta k + 1} }, 
    \end{displaymath}
    and 
    \begin{displaymath}
    q_{r}(k,\beta, x_R) = \frac{1}{2} +  \frac{m}{4} \frac{1-\sigma}{1-\sigma+ \sigma (1-x^{*}_{R}(m))^{\beta k + 1}} 
    \end{displaymath}
    that along side with the ideological center point define what I call 'echo chambers': an enclosed system in which information is truthfully and credibly shared and outside which no credible information is transmitted. 
    
    \noindent A voter whose ideological bliss point is between the cut-offs $ (q_l (k,\beta, x_R), 1/2 )$ (resp. $ (1/2 , q_r (k,\beta, x_R) )$) truthfully and credibly transmit information only to her peers leaning toward party $L$ ($R$), and do not convey any information otherwise.\footnote{It can be shown that this result will also hold if network affiliation were costly. To see this, consider the case in which an ex-post independent voter has to decide if and which network to join at a cost $ \alpha >0$. From Theorem 1 it follows that he rely on the information received from voters with aligned biases that belong to his group type and he disregards all the information received by others. Then, if the degree of homophily is high enough, an independent voter would be willing to incur a cost to join a network, which, in turn, reinforces the echo chambers effect. Note also that an ex-post partisan would never pay to buy information that she is not willing to use.} 
    
    It is easy to prove that the two cut-offs $q_{l}(k,\beta, x_L)$ and $q_{r}(k,\beta, x_R)$ are, respectively, decreasing and increasing in the degree of homophily $ \beta$, in the sources of receiver information $k$, and in the advertising level $x_J (t_J)$, which imply that as the number of reliable sources of information increases the echo chambers become wider, {\em i.e.} the distance $\mid q_j(k,\beta, x) - 1/2 \mid$ for each $j \in \{l , r\}$ increases.

    \section{Echo Chambers and Parties' Advertising Strategies}
    
    I assume parties can advertise their candidate either by using a random technology, which allow each voter to be exposed at the advertisement with the same positive probability, or by relying on a group targeting technology, which direct the advertisement only to voters who are either $L$-leaning or $R$-leaning. 
    I compare the maximum unit cost that a party is willing to pay to advertise its candidate under both technologies and I show how those choices are affected by the echo chambers. 
    
    \noindent I show that if the voters' network is characterized by many (resp. few) source of credible information, \emph{i.e.} $k \beta $ is large (small), parties exploit the echo chambers by advertising their candidate randomly (to targeted groups).

 	\subsection{Random advertising}
 	Suppose party $L$ has a moderate candidate and it randomly advertises him. Then, the proportion of voter informed about $L$'s candidate type is
 	$$ \gamma_{L} (x_{L}, k, \beta) = 1 - \big(1 - x_{L}(m)\big)^{ \beta k + 1},$$
 	which is increasing in the number of information sources of the voters $k$, in the degree of homophily in the network $ \beta $, and in the level of advertising $x_J(m)$. 
 	
    When both parties have a moderate candidate, given they play symmetric strategies, each party wins the election with probability $\pi(m,m\mid x, k, \beta) = 1/2$. If party $R$ has an extremist candidate what matter is the distribution of indifferent voters which depends on the information set held. Then, for any possible information set $I$, the votes share of party $L$ is
    \begin{displaymath}
    \mu( x \mid \theta, k, \beta )=  \frac{1}{2} + \frac{m}{4} \rho( m, e \mid I,  k, \beta) \gamma_{L} (x_{L}, k, \beta), 
    \end{displaymath}
    which is increasing in $\gamma_{L} (x_{L}, k, \beta)$. Thus, when the voters' personal network features a high degree of homophily ($\beta$) and/or the network richness ($k$) is high, the echo chamber are wide and advertisement is echoed among voters.\footnote{Wider echo chambers refers to the two cut-offs in Theorem \ref{Echo}, as $\beta$ or $k$ or both increase the two cut-offs $(q_l(k,\beta, x_L), q_r(k,\beta, x_R))$ moves toward the extremes of the ideological line making the echo chambers wider, {\em i.e.} credible information travels among an higher number of voters.}  Conversely, if the voters' personal network features a low degree of homophily and/or a small number of peers, the echo chambers are small preventing the diffusion of information among voters. As result, many voters remain uninformed and they will vote for their favoured party.

 	\begin{theorem}\label{Random}
    In a symmetric pure strategy equilibrium there exist critical levels of richness of the network $k$ and degree of homophily $\beta$, such that party $J$ randomly advertises a moderate candidate if and only if 
 	\begin{equation*}
    c(k, \beta) < c^* (k, \beta) \equiv \frac{ (2 -3 m) (1 - \sigma) (\beta k + 1)(1- \rho(m, m \mid I, \beta, k)) }{16} ,
 	\end{equation*}
 	and never advertise a moderate candidate otherwise.
    \end{theorem}
 	\begin{proof}
 		See Appendix.
 	\end{proof}

    \subsection{Targeted Advertising}
    The new informational environment has made available a large amount of voters' personal data to political actors: voters/users to join social media create profiles in which they specify personal information like gender, age, education, jobs, and beliefs. Access to this data virtually allows parties to target their advertisement and manipulate the flow of information diffusion exploiting the presence of echo chambers. 
    In this section I find the conditions under which parties would rely on a group targeting advertisement technology. I assume that if a party targets its advertisement to voters biased toward party $J$ then all those voters receive the advertisement with probability one.
    
    \begin{theorem}\label{target}
	In a symmetric pure strategy equilibrium parties never target their advertisements to their own supporters, and there exist critical levels of richness of the network, $k$, and degree of homophily, $\beta$, such that party $L$ advertises a moderate candidate exclusively to the opponent's supporters if and only if the cost of advertisement is such that 
		\begin{equation*}
	c(k, \beta) < \hat{\bar{c}} \equiv \frac{(2-3m -\sigma m)}{4}.
		\end{equation*}
	\end{theorem}
	\begin{proof}
		See Appendix.
	\end{proof}
	
    Given that parties never advertise an extremist, to see the intuitions behind Theorem 3, as before, it is sufficient to focus on the case in which the candidate of party $L$ is a moderate. 
 	
    First, I focus on the first part of the statement: parties never exclusively target the advertisement on their own supporters.
    From Theorem \ref{Echo} follows that the information does not travel between left and right leaning voters even if they belong to the group of ex-post independent voters. Thus, if party $L$ targets the advertisement to its own supporters all the right leaning voters remain uninformed about its own candidate, and consequently they vote for party $R$. Thus, party $L$ wins the election with probability $1/2$, regardless of the candidate type of its opponent. If party $L$ randomly advertise its candidate whenever it faces an extremist opponent it wins the electoral competition with probability $ \pi_{L}(m,e \mid x) = 1/2 + (1- \rho(m,m \mid I, \beta, k)) (1- (1-x_{L})^{(\beta k + 1)} > 1/2$, while if it faces a moderate advertised randomly, then it wins with probability $1/2$. In other words, by targeting the advertisement to its own supporters, a party wastes its persuasion power on those voters who would have voted for it even if they were not hit by the advertisement. 
 	
    The second part of Theorem \ref{target} links the network structure, which affects information diffusion, to parties' choice of advertisement technology, {\em i.e.} under which conditions parties prefer to target their advertisement to the voters leaning toward their opponent (opponent's supporters) rather than randomly advertise their candidate to all voters. 
    Consider the case in which party $L$ has a moderate candidate and it targets the advertisement to the opponent's supporters, while party $R$ randomly advertise its candidate whenever it has a moderate. 
    By targeting its advertisement to $R$'s supporters all the ex-post left independent voters will be uninformed about their favoured party's candidate ($t_L = \emptyset$), while every ex-post right independent biased voter will hold the information $ t_{L} = m $. 
    
    It follows that $R$'s candidate is an extremist, party $L$ wins the electoral competition with probability one. As both ex-post $R$-leaning independent voters  - with information $(m, M_L(m), \emptyset, M_R(\emptyset))$ - and the uninformed ex-post $L$-leaning independent voters vote for party $L$.
    If, instead, party $R$ has a moderate and advertise him randomly then a proportion of ex-post independent voter 
    $ \gamma_R (x_R , k, \beta)= 1- (1-x_{R}(m))^{\beta k + 1} $ vote for party $R$. As a consequence, by targeting the opponent supporters ($ x_{LR}^d (m)$), party $L $ wins the election with probability $\pi^{d}_{L}(m,m \mid x_L^{d},x_R) = 1/2 - ( (1- \rho (m,m \mid I, \beta, k)) (1-2(1- x_{R})))/16 < 1/2 =  \pi_{L}(m,m \mid x) $, where $\pi_{L}(m,m \mid x) $ is the probability that party $L$ wins the election if it would have randomly advertised a moderate. 
    
    Theorem \ref{Random} and \ref{target} highlight the relationship between the advertisement's cost and the network structure. The wider (resp. small) the echo chambers are the lowest (higher) is the parties' incentive to target the supporters of their opponent. 
    
    It is easy to show that $\hat{\bar{c}}(k, \beta) < c^* (k, \beta)$ for $k \beta \geq \bar{k \beta}$. When voters' network is characterized by high degree of homophily and/or the richness of the network it acts as an information diffusion device.\footnote{Where {\small
    $\bar{k \beta} \equiv \frac{(2-3m)\left( (2+(1-
\rho(m,m \mid I, \beta, k))\sigma)+(1+\rho(m,m \mid I, \beta, k))\right) -4m \sigma)}{(2-3m)(1-\rho(m,m \mid I, \beta, k))(1- \sigma) }$ } (see Appendix for derivation).} In turn this implies that by randomly advertising their candidate parties can reach many echo chambers, which spread the information, and allow parties to maximize the probability of winning the election for any given level of advertisement $x_j(m)$.
    On the contrary, when the voters' network is characterized by low homophily and/or the richness of the network is small, $k \beta <  \bar{k \beta}$, the network's structure entails a significant waste of information. This implies that the proportion of voters informed about the $R$'s candidate type, $ \gamma_R(x_R , k, \beta)$, is small and party $L$'s probability to win the election by targeting the $R$-leaning voters is high. 
 	
    My results are robust to the introduction of heterogeneity between the two ex-ante bias groups of voters. Let the degrees of homophily within the left and right group of independent voters be, respectively, $\beta_l < \beta_r $. From Theorem \ref{Echo} follows that the $R$-leaning voters have wider echo chambers ($ | 1/2 - q_l (\beta_l, k) | < | 1/2 -  q_r(\beta_r, k) |$) which implies, ceteris paribus, that $R$-leaning ex-post independents would have access to more credible sources of information than the ex-post $L$-leaning independents have. 
 	
    Thus, party $L$ has two advantages: first, the ex-post left partisan group will be wider than the ex-post right partisans one; second, party $L$ needs less advertising than the right party in order to persuade the same proportion of $R$-leaning ex-post independent voters with respect to that party $R$ needs to persuade the $L$-leaning ex-post independent voters.
    This setting mimics a scenario in which the two parties incur different costs to randomly advertise their candidate. 
    From Theorem \ref{Random} and \ref{target} follows that party $L$ will randomly advertise its moderate candidate and exploit the echo chambers without jeopardizing the support of its own supporters.

 	\section{Discussion}
 	
    The preceding analysis aimed to disentangle the mechanisms by which communication through social media networks shapes electoral competition.
    I designed a model to highlight the structure of the problem in its simplest possible form, but the main results do not hinge on some of the specific assumptions I made. For example, both the symmetry in candidates types and in the voters' network structure are convenient but not essential.  
    I have also abstracted from negative advertising: allowing parties to spread truthful negative advertising about the candidate's type of their opponent translate in a model with a larger set of signals, but the structure of the problem remains unchanged.\footnote{Although part of the political content on social media is negative information about competitors (produced by candidates rather than parties' elites), it has been proven by \cite{NegAdV_response} and \cite{negAdEmp} that negative advertising elicits negative effects for both the target and the sponsor. \cite{negADV} show that negative advertising mobilizes partisans, but depresses turnout among independents.
    }
    
    Another important assumptions in the model is that voters network each voter knows exactly the ideology of the other voters in her network. This implies that, in equilibrium, a voter will be fully convinced that a candidate is moderate as soon as she either receives party's advertisement or a single cheap-talk message from someone who belongs to her echo chamber. That is definitely not true once one consider the possibility of preference uncertainty in the voters' network.
	In this case voters would need to look at the distribution of messages, as one message is not going to be sufficient, and the parties might need to invest more on advertisement to convince a large enough mass of voters.
    
   
   

    \subsection{Strategic candidate's selection}
    In what follows I briefly analyse the case in which parties can strategically select their candidate. 
    I show that if advertisement technology is inefficient parties select an extremist candidate (and never advertise him), otherwise parties randomize between the two types of candidate and advertise only a moderate one. 
    To simplify the argument I assume that parties may only randomly advertise their candidates. 
    Let $s_J= (\sigma_J(t_J),x_J(t_J))$ be the strategy of party $J$, where $\sigma_J(t_J)$ is now the probability that party $J$ selects a candidate of type $t_J \in \{e, m\}$.\footnote{The Theorem \ref{stratCand} is an easy corollary of the Proposition 1 of \cite{galeottiMatt_personal_2011}, in which $\overline{c}(k, \beta) = \frac{(2-7m) (1+ \beta k)}{16(1- k \beta (1- \underline{p})}$ and $\underline{p}$ solves $\frac{16 c}{2-3m} =  (1+ \beta k) \underline{p}^{\beta k}, $ with  $\underline{p} = 1-x$. }
    \begin{theorem}\label{stratCand}
    For every $k$ and $\beta$ there exist a critical level of $\overline{c}(k, \beta)$ such that in equilibrium:
    \begin{itemize}
        \item[i)] If $c \geq \overline{c}(k, \beta)$ both parties always run with an extremist and do not advertise him.
        \item[ii)] If $c < \overline{c}(k, \beta)$ parties randomize between moderate and extremist and $\sigma^*$ and $x^*= x^*(m)$ solve
        \begin{eqnarray}
    &  &  (1- \sigma) (\beta k + 1)  \left(1- \rho (m , m \mid I , k, \beta) (1-x)^{\beta k} \right)=
       \frac{16 c}{2 -3 m} \\
      &  &    1- \rho (m , m \mid I , k, \beta) \left(1- (1-x)^{\beta k+1}\right) = \frac{ 4 m + 16 c x^*}{2-3m}.
        \end{eqnarray}
    \end{itemize}
    
    \end{theorem}

    First consider the symmetric equilibria in pure strategy. 
    Suppose parties always select a moderate, $t^* = m$. Note that in any pure strategy equilibria it must be that $x^*(t^*) = 0$, as voters will anticipate parties candidate's choice. Then, if a moderate is selected in a pure-strategy equilibrium we have that $x^*(m) = 0$, which implies that the parties have incentive to deviate and select an extremist. In fact, by selecting an extremist a party increases its utility, as it is policy motivated, without harming the probability of winning the electoral competition as voters, which in equilibrium are uninformed, would believe that both parties are running with a moderate.
    
    Now consider the case in which parties always select an extremist,  $t^* = e$. 
    If a party deviates by selecting a moderate and advertising him there are two effects to consider. First, there is a positive effect on the party's expected utility: by advertising a moderate the party increases its probability of winning the elections. As independent voters prefers to vote a moderate candidate, informing them of the deviation does increase the party's probability of winning the election. Second, there is a negative effect on party's expected utility, the advertisement is costly. Thus, if the cost of advertisement is too high $c \geq \overline{c}(k, \beta)$, {\em i.e.} the marginal cost of the deviation (the cost of advertisement) is higher than its marginal benefit (the increase in the probability of winning the election), parties always select an extremist candidate and do not advertise him. 
    
    
    When the advertisement technology is efficient, choosing a moderate increases the probability of winning the election and compensate the cost of advertisement. But, as parties are policy motivated, they would still prefer to win with an extremist candidate rather than with a moderate.\footnote{This also implies that, ceteris paribus, as more centrist is the moderate candidate's position the less often a moderate candidate is selected by the party.} 
    Then, whenever the cost of the advertisement is not too high parties randomly select a candidate type and advertise only the moderate one.

    	\section{Conclusion}
    	
    Recent technological changes have lead to the rapid development of communication networks through social media, and its political implications are far from being completely understood. This paper poses the basis to understand how social interaction among voters with heterogeneous preferences influence the information diffusion and shapes the political outcome.  
    I build a formal model that embeds \emph{(i)} a model of campaign competition, where parties compete in campaign advertising, and \emph{(ii)} a model of personal influence, where voters, with possibly heterogeneous preferences, can strategically communicate with each other to affect the political outcome. 
 	
    The first result shows that whenever voters can strategically communicate with each other echo chambers arise endogenously, voters share valuable information only with like-minded peers. In other words, a chamber arises if and only if the senders and the receiver are biased toward the same party, {\em i.e.} small ideological distance between senders and receivers is not enough to guarantee truthful communication, as one would have expected. Although small ideological distance between the sender and the receiver is necessary for truthful communication it is not sufficient, voters need to be leaning toward the same party to create an echo chamber.
    The second result suggests that both the richness of the network and the degree of homophily play an important role in determining the advertising strategies of the parties. Specifically, whenever one of the two or both are low, the parties are likely to target their advertising campaign to voters ideological biased toward their opponent. This is in sharp contrast with the literature with truthful communication, which suggests that parties seek to target their ideologically closer voters.
    
    From a more general perspective, my results suggest that the role of social networks in political competition should not be underestimated, and neither should the role-played by the voters within their network. Understanding how social networks shape voting sharing strategies is crucial to stem the phenomenon of ``fake news'' and eventually identifying the incentives of the nodes that create and spread them.

    \appendix
    
    \section*{Appendix }
 	This Appendix contains proof of all results stated in the text.
 	
 	 	\paragraph{The indifferent voter.} Given an information set $I$, a voter $i$ is indifferent between voting for the two parties if and only if 
 	\begin{equation}\label{Aindiff}
    u_i(v_i =1  \mid \theta, I, k) =  u_i(v_i=0 \mid \theta, I, k)
    \end{equation}
 	where 
 	\begin{small}
 	\begin{eqnarray}\label{Aind1}
 	 u_i(v_i =1  \mid \theta, I, k) &=&  \sum_{t_{L}, t_{R} \in T}  \rho ( t_{L}, t_{R} \mid I, x, k)  (t_L - i) 
 	\end{eqnarray}
 	\begin{eqnarray}\label{Aind0}
 	  u_i(v_i =0  \mid \theta, I, k) &=& \sum_{t_{L}, t_{R} \in T} \rho ( t_{L}, t_{R} \mid I, x, k)  (i -(1-t_R))  
 	\end{eqnarray}
 	\end{small}
    To ease the exposition I denote the candidate type of party $J$ as $t_J= \{e_J, m_J\}$. 
    Given that $ \sum_{t_{L}, t_{R} \in T}  \rho ( t_{L}, t_{R} \mid \theta, I, k) = 1 $, the  (\ref{Aind1}) becomes
     \begin{small}
 	\begin{eqnarray}
     u_i(v_i =1  \mid I, x, k) &=& \big( \rho(m,m \mid I, x, k)  + \rho(m, e \mid I, x, k)  \big) m_L \notag \\
     &+& \big( \rho(e,m \mid I, x, k)  +\rho(e,e \mid I, x, k)  \big) e_L - i \notag
     \end{eqnarray}
 	\end{small}
    and (\ref{Aind0}) becomes
     \begin{small}
 	\begin{eqnarray}
    u_i(v_i =0  \mid \theta, I, k) &=& - \big( \rho(m,m \mid I, x, k)  + \rho( e,m \mid I, x, k)  \big) m_R \notag \\
 	 &-& \big( \rho(m,e \mid I, x, k)  +\rho(e,e \mid I, x, k) \big)  e_R +i - 1  \notag
    \end{eqnarray}
 	\end{small}
 	Thus, by substituting (\ref{Aind1})  and (\ref{Aind0}) into (\ref{Aindiff}), after some manipulation,
 	 \begin{small}
 	\begin{eqnarray}
    1 - 2 i  &=& \rho(m,m \mid I, x, k) (m_R -m_L) + \rho( e,e \mid I, x, k)  (e_R - e_L) \notag \\ 
    &+&  \rho(m,e \mid I, x, k) (e_R - m_L) + \rho(e, m \mid I, x, k) (m_R - e_L) \notag
    \end{eqnarray}
 	\end{small}
  	from which it is immediate to obtain
    (\ref{indifferent}). 
\hfill $\blacksquare$
\bigskip
 	
 	\paragraph{Sketch of proof of Lemma \ref{lemcandidate}.}  To simplify the argument consider the case in which party $L$'s candidate is an extremist advertised with intensity $ x_{L}(e) > 0 $, and there is only one voter $i$.
    When voter $i$ is informed only about the party $L$'s candidate, {\em i.e. } her information set is $t_L = m$ and $t_R = \emptyset$, she will vote for the left party if and only if her bliss point is such that $ i < \frac{1}{2} - \frac{m}{4} \rho(e,m \mid (e, \emptyset), x) \equiv \underline{i}$. If, instead voter $i$ is uniformed she votes for the party $L$ if and only if her bliss point is such that $ i < \frac{1}{2} \equiv \overline{i} $. Given that $\overline{i} > \underline{i}$, party $L$ by advertising its extremist candidate, not only pays a cost $c>0$, but it also lowers his-own probability of winning the election. 
\hfill $\blacksquare$
\bigskip

 	\paragraph{Sketch of proof of Lemma \ref{alphak}.}  
 	
    An ex-post partisan voter is a $J$-leaning voter who, regardless the information she holds on the candidates' type, always votes for the party she is ex-ante biased for, {\em i.e.} $J$. An ex-post $J$-leaning independent voter is a voter who votes for the opponent party, $J^{-}$, whenever he knows that $J^{-}$'s candidate is a moderate and has no information about the candidate of the party he is ex-ante biased for.  
 	
 	\noindent Consider a right-leaning voter $i \in (1/2, 1)$ and $\rho (m,m \mid \theta, I, x, k )$ be the voter $i$ posterior beliefs that both candidate are moderate given his information set $I$, the parties strategies $x$, and his network $k$. From Definition \ref{D1} follows that a right leaning voter is an ex-post right-leaning independent voter if and only if she votes for the left party whenever his information set $(t_L, M_L(t_L), t_R, M_R(t_R))$ is either $(m, \emptyset, \emptyset, \emptyset)$ or $(\emptyset, m, \emptyset, \emptyset)$, that is the following inequality holds
 	\begin{equation}\label{voterdef}
 	u_i(v_i=1 \mid \theta, I, x, k) \geq u_i (v_i =0 \mid \theta, I, x, k), 
 	\end{equation}
    where $v_i \in \{0,1\}$ is $i$' voting decision and $v_i =1 $ denotes the decision of voting for party $L$.
 	The inequality (\ref{voterdef}) holds whenever 
 	\begin{equation*}
 	i \leq \dfrac{1}{2} + \dfrac{m}{4} \left[ 1 - \rho (m,m \mid \theta, I, x, k)\right] \equiv \alpha_r.
 	\end{equation*}
 	Thus, all voters whose bliss point are such that $i \in (1/2 , \alpha_r)$ are ex-post right-leaning independent voters. Given the symmetry of the game by the same argument it is possible to show that all voters whose bliss point 
 	is such that $i \in (\alpha_l, 1/2 )$ are ex-post left-leaning independent voters, where $ \alpha_l = \dfrac{1}{2} - \dfrac{m}{4} \left[ 1 - \rho (m,m \mid I, x, k)\right] $. 
 	Putting together the above results, a voter is ex-post independent whenever his ideological bliss point $i$ belongs to the interval $ [\alpha_l, \alpha_r]$, and he is an ex-post partisan otherwise.
\hfill $\blacksquare$

\bigskip

    \paragraph{Proof of Lemma \ref{paircheaptalk}.} The proof of point (i) of Lemma 1 is trivial. To prove the remaining two points I divide the proof in two parts. I first analyse the voting decision of a receiver who believes the sender shares only truthful messages. Then, I prove the result (ii) and (iii) of Lemma \ref{paircheaptalk}. 
 	
 	\textbf{Part \emph{(1)}:} \ If the ex-post independent receiver knows his favoured party has a moderate candidate then, regardless of the information he holds on the opponent's candidate, she casts her vote for her favoured party.
 	
 	\noindent An ex-post independent receiver credibly and partially informed only about $L$'s candidate type,{\em i.e.} who holds an information $ I^r = \{ (m, M^s_L(\emptyset), \emptyset, M^s_R(\emptyset)); (m,M^s_L(m), $ $\emptyset, M^s_R(\emptyset)); (\emptyset, M^s_L (m), \emptyset, M^s_R(\emptyset)) \}$, votes for party $L$ if and only if
 	$u_{r}(1 \mid I^{r}, x, s) > u_{r}(0 \mid I^{r}, x, s)$, which holds for $r < \dfrac{1}{2} + \dfrac{m}{4} \left[ 1 - \rho (m,m \mid I^{r}, x, s)\right]$.
 	Mutatis mutandis, an ex-post independent receiver votes for party $R$ when credibly and partially informed only about $R$'s candidate type  if and only if $ r > \dfrac{1}{2} - \dfrac{m}{4} \left[ 1 - \rho (m,m \mid I^{r}, x, s)\right]$.
 	
 	\noindent The uninformed ex-post independent receiver $ I^{r} =(\emptyset, M_L^s( \emptyset),$  $ \emptyset, M_R^s(\emptyset)) $ always casts his vote for her favoured party. 
 	
 	\textbf{Part \emph{(2)}:} \ The sender behaves as if his receiver is pivotal. Then, truthful revelation is incentive compatible if and only if
 	\begin{equation*}
 	u_{s}((M_{L}(t_{L}),M_{R}(t_{R})) \mid I^{s}, x , r) \geq u_{s}(\hat{M}_{L}(t_{L}),\hat{M}_{R}(t_{R})\mid I^{s}, x , r).
 	\end{equation*}
 	
 	\noindent Assume that the sender is matched with an ex-post left-leaning independent receiver \emph{i.e.} $ r \in \left( \frac{1}{2} - \frac{m}{4} \frac{1 - \sigma_{L}}{\sigma_{L} (1-x_{L})^{k+1} + 1 - \sigma_{L}} , \frac{1}{2} \right)$. 
 	
 	\noindent First, consider a perfectly informed sender, $ I^{s} = (m,m) $. The sender's expected payoffs from every possible message pair $M_{L}(t_L), M_{R}(t_R)$ are:
 	\begin{small}
 		\begin{eqnarray*}
 			u_{s}(m, m \mid  I^{s}, x , r) & = &  u_{s}((m, \emptyset) \mid I^{s}, x , r) =  m - s, \\
 			u_{s}(\emptyset, m \mid I^{s}, x , r) & = &  x_{L} (m - s)  (1 - x_{L}) (-1 + m + s), \\
 			u_{s}(\emptyset, \emptyset \mid I^{s}, x , r) & = &  (1 - x_{R}(1 - x_{L})) (m - s) +  ( x_{R} (1 - x_{L})) (- 1 + m + s). 
 		\end{eqnarray*}
 	\end{small}
 	Then, the sender truthfully shares his information if and only if $u_{s}(m, m \mid \theta, I^{s}, x , r)  \geq u_{s}(\emptyset, m \mid \theta, I^{s}, x , r) $
 	and $u_{s}(m, m \mid \theta, I^{s}, x , r)  \geq u_{s}(\emptyset, \emptyset \mid \theta, I^{s}, x , r) $, that hold if and only if $ s \leq  1/2$. A perfectly informed sender truthfully reveal his information if and only if the receiver is biased toward his favoured party, regardless whether the receiver is an ex-post partisan or an ex-post independent voter. Clearly, the same is true whenever the sender is only informed about the party's candidate type he prefers.\\
 	
 	\noindent When the sender holds the information set $I^{s} = (\emptyset , m)$, the sender's payoff from each possible message pair are:
 	\begin{small}
 		\begin{eqnarray*}
 			u_{s}(\emptyset, m \mid  \theta,  I^{s}, x, r )   &=&   \rho (e,m \mid  \theta, I^{s},x, r)   (-1+m+s)  \\
 			& + &  \rho (m,m \mid  \theta, I^{s}, x, r)  [  x_{L} (m - s) + (1 - x_{L}) (-1 + m + s) ] ;  \\
 			u_{s}(m, \emptyset \mid  \theta, I^{s}, x,r)  & = &  u_{s}(m, m \mid I^{s}, x , r)\\
 			& = &  \rho (m,m \mid  \theta, I^{s}, x, r)  (m-s)  + \rho (e,m \mid  \theta, I^{s}, x, r)   (m/2-s); \\
 			u_{s}(\emptyset, \emptyset \mid  \theta, I^{s}, x, r)  & = &   \rho (m,m \mid  \theta, I^{s}, x, r) [ x_{L} (m-s)\\
 			& + & ((1-x_{L})x_{R})(-1+m+s) + (1 - x_{R}(1 - x_{L})) (m - s) ] \\
 			& + &  \rho (e,m  \mid  \theta, I^{s},x, r)  [(1-x_{R}) (m/2-s) + x_{R}(-1+m+s) ].
 		\end{eqnarray*}
 	\end{small}
    The sender never lies by revealing false information about the candidate's type of the left party. When the sender discloses false information about the left candidate $(m, M_R(t_R))$ then the sender knows the receiver uses his message to cast his vote for the left party and with probability $ (1 - \sigma_{L}) $ will take the wrong decision by voting for an extremist candidate. Thus, the only two message pairs to consider are $  ( \emptyset, m )$ and $(  \emptyset, \emptyset )$. It follows that the sender truthfully shares his information if and only if $ u_{s}(\emptyset, \emptyset \mid  \theta, I^{s}, x, r) \geq u_{s}(\emptyset, m \mid I^{s}, x, r)$ which holds if and only if $	s \geq \frac{1}{2} - \frac{m}{4} \frac{1-\sigma_{L}}{\sigma_{L} (1-x_{L})^{2} + 1 - \sigma_{L}}$.
 	
    An ex-post partisan sender informed only about the candidate of his favoured party always truthfully shares his information. On the contrary, if an ex-post partisan sender is informed only about the candidate of the party he likes less, he never truthfully shares his information. This proves Lemma 1 \emph{(ii)}. \\
 	
    \noindent A $J$-leaning ex-post independent sender always truthfully shares his information, regardless the information hold, if also the receiver is $J$-leaning, {\em i.e.} the preferences of the voters in the matched pair are aligned (Lemma \ref{paircheaptalk} (iii)) . 
 	
    \noindent  When the sender has no information about neither of the candidate type of the two party (he is uninformed) his expected payoffs from each messages pair are:
 	\begin{small}
 		\begin{eqnarray*}
 			u_{s}(m, m \mid \theta, I^{s}, x, r)  & = & u_{s}(m, \emptyset \mid \theta, I^{s}, x, r) =  \\
 			& = &  (m - s) (\rho(m,m \mid \theta, I^{s}, x, r) + \rho(m,e \mid \theta, I^{s}, x, r) )   \\
 			& + & (m/2 - s) (\rho(e, m \mid \theta, I^{s}, x, r) + \rho(e,e \mid \theta, I^{s}, x, r) ); \\
 		u_{s}(\emptyset, m \mid \theta, I^{s}, x, r) & = & \rho(m,m \mid \theta, I^{s}, x, r) (x_{L}(m -s) + (1-x_{L}) (-1+m+s) ) \\
 			& + &   \rho(m,e \mid \theta, I^{s}, x, r) (x_{L} (m -s) +  (1-x_{L}) (-1+ m/2 +s) ) \\
 			& + &   \rho(e, m \mid \theta, I^{s}, x, r) (-1+m+s) +  \rho(e,e \mid \theta, I^{s}, x, r) (-1+ m/2 +s);\\
 			u_{s}(\emptyset, \emptyset \mid \theta, I^{s}, x, r) & = & \rho(e,e \mid \theta, I^{s}, x, r)(m/2 - s) + \rho(m,e \mid \theta, I^{s}, x, r) (m-s)\\
 			& + &\rho(m,m \mid \theta, I^{s}, x, r) \left( (1-x_{R}(1-x_{L}) ) (m-s)  +  x_{R}(1-x_{L} ) (-1+m+s) \right) \\
 			& + &  \rho(e,m \mid \theta, I^{s}, x, r)  (x_{R}(-1+m+s) + (1-x_{R}) (m/2 - s)).
 		\end{eqnarray*}
 	\end{small}
 	
 	\noindent The sender truthfully share his information if and only if the following two conditions hold:
 	\begin{itemize}
 	    \item[i)] the sender is not an ex-post partisan voter leaning for the party from which he did not receive any information, $s \geq \frac{1}{2} - \frac{m}{4} \frac{1-\sigma_{L}}{\sigma_{L} (1-x_{L})^{2} + 1 - \sigma_{L}}$,
 	    \item[ii)] the sender's preferences are aligned with his receiver's preferences $ s < 1/2$.
 	\end{itemize}
 \hfill $\blacksquare$
 	\bigskip

 	\paragraph{Proof of Lemma \ref{lemSeq}.} This result follows from sequential rationality and consistency requirements. 
 	\hfill $\blacksquare$

\bigskip

 	\paragraph{Proof of Theorem \ref{Echo}.} \ To show that a senders never deviate from the message strategies described in the two-players game, even with richer the communication network, it is sufficient to show that there is not profitable deviation. Let $ (k-1) $ senders follow the equilibrium strategies of the two-player cheap talk game described before, and  analyse the incentive to deviate of the $k^{th}$ sender, with information $ I^{s} $.\footnote{This is without loss of generality, given that within the sampled network of each voter there is no preference uncertainty and the communication among voters has the form of private messages. Thus, each voter evaluates the messages she receives one by one and can disregard all message pairs received from  senders biased for the party she likes less.} Denote by $\beta \in (0, 1)$ the probability that in each of the $k$ matched pair both voters have aligned preferences, they both prefer the same party. 
 	
 	\noindent To simplify the exposure I analyse the case in which the sender is matched with a left leaning receiver. 
 	
 	\noindent The expected payoffs of a perfectly informed sender, $ I^{s} = (m,m) $, from every messages pair are:
 	\begin{small}
 		\begin{eqnarray*}
 			u_{s}(m, m \mid I^{s}, x, k, \beta) & = & u_{s}(m, \emptyset \mid I^{s}, x, k, \beta)  =  m - s \\
 			u_{s}(\emptyset, m \mid I^{s}, x, k, \beta)  & = & x_{L}^{\beta (k-1) +1} (m - s) + (1 - x_{L})^{\beta (k-1)+1} (-1 + m + s) \\
 			u_{s}(\emptyset, \emptyset \mid I^{s}, x, k, \beta)  & = &  \left((1 - x_{R})(1 - x_{L})\right)^{\beta (k-1)+1} (m - s)\\
 			&+& \left(( x_{R}(1 - x_{L}))\right)^{\beta (k-1)+1} (- 1 + m + s) 
 		\end{eqnarray*}
 	\end{small}
    and a perfectly informed sender sends a truthful message to a left biased receiver if and only if $u_{s}(m, m \mid I^{s}, x, k, \beta) \geq u_{s}(\emptyset, m \mid I^{s}, x, k, \beta)$ and $u_{s}(m, m \mid I^{s}, x, k, \beta)  \geq  u_{s}(\emptyset, \emptyset \mid I^{s}, x, k, \beta)$. The two inequalities hold if and only if $ s \leq 1/2$. 
    A perfectly informed sender has no incentive to deviate from the equilibrium strategy of two-players communication game: he truthfully shares his information whenever the his preferences are aligned with the receiver, regardless of the ex-post group type (partisan/independent) each of them belongs.
 	
 	\noindent The expected payoffs of a sender only informed about the candidate type of the receiver's less favoured party, $I^{s} = (\emptyset , m)$, from each messages pair are:
 	\begin{small}
 		\begin{eqnarray*}
 			u_{s}(\emptyset, m \mid I^{s}, x, k, \beta) & = & \rho(m,m \mid I^{s},  k, \beta) \left(  x_{L}^{\beta (k-1)+1} (m - s) + (1 - x_{L})^{\beta (k-1)+1} (-1 + m + s) \right)  \\
 			& + &  \left(1- \rho(m,m \mid I^{s},  k, \beta) \right) (-1+m+s) \\
 			u_{s}(\emptyset, \emptyset \mid I^{s}, x, k, \beta) & = & \rho(m,m \mid I^{s},  k, \beta) \left( x_{L}^{\beta (k-1)+1} (m-s) + ((1-x_{L}) x_{R})^{\beta (k-1)+1} (-1+m+s) \right. \\
 			& + &	\left. ((1 - x_{R})(1 - x_{L}))^{\beta (k-1)+1} (m/2 - s) \right)  \\
 			& + & (1 - \rho(m,m \mid I^{s},  k, \beta)  ) \left( x_{R}^{\beta (k-1)+1} (-1+m+s) + (1-x_{R})^{\beta (k-1)+1} (m/2-s)\right)
 		\end{eqnarray*}
 	\end{small}
 	A sender truthfully shares his information if and only if $u_{s}(\emptyset, m \mid I^{s}, x ,k, \beta )  \geq u_{s}(\emptyset, \emptyset \mid I^{s}, x ,k, \beta )$, that is if and only if $s \geq \frac{1}{2} - \frac{m}{4} \frac{1-\sigma_{L}}{\sigma_{L} (1-x_{L})^{\beta k + 1 } + 1 - \sigma_{L}} \equiv \overline{s}_l$. Thus, a sender partially informed about the receiver's less favoured candidate never deviates from the equilibrium strategy prescribed in the two-player cheap talk game. 
 	
 	\noindent The expected payoffs of an uniformed sender, $ I^{s} = (\emptyset, \emptyset) $, from each messages pair are:
 	\begin{small}
 		\begin{eqnarray*}
 			u_{s}(m, m  \mid I^{s}, x ,k, \beta ) & = &  u_{s}((m, \emptyset) \mid I^{s}, x ,\theta ) \\
 			& = &  (\rho(m, m \mid I^{s},  k, \beta) + \rho(m, e \mid I^{s},  k, \beta) )  (m - s) \\ 
 			& + & (\rho(e, m \mid I^{s},  k, \beta)+ \rho(e, e \mid I^{s},  k, \beta)) (m/2 - s), 	\\			
 		    u_{s}(\emptyset, m  \mid I^{s}, x ,k, \beta ) & = & \rho(m, m \mid I^{s},  k, \beta) (x_{L}^{\beta (k-1)+1} (m -s) + (1-x_{L})^{\beta (k-1)+1} (-1+m+s) ) \\
 			& + & \rho(m, e \mid I^{s},  k, \beta) \left(x_{L}^{\beta (k-1)+1} (m -s)+(1-x_{L})^{\beta (k-1)+1} (-1+ m/2 +s) \right)\\
 			& + & \rho(e,m\mid I^{s},  k, \beta)  (-1+m+s) + \rho(e,e\mid I^{s},  k, \beta)  (-1+ m/2 +s),\\
 		u_{s}(\emptyset, \emptyset \mid I^{s}, x , k, \beta) & = & \rho(m,m\mid I^{s},  k, \beta)  \left( \left( ((1-x_{L}) (1-x_{R}))^{\beta (k-1)+1} + x_{L}^{\beta (k-1)+1} \right) (m-s) \right.\\
 			& + &  \left.   (x_{R}(1-x_{L}))^{\beta (k-1)+1} (-1+m+s) \right)  \\
 			& + &\rho(e,m\mid I^{s},  k, \beta) \left(x_{R}^{\beta (k-1)+1} (-1+m+s) + (1-x_{R})^{\beta (k-1)+1}(m/2 - s)\right) \\
 			& + & \rho(e,e\mid I^{s},  k, \beta) (m/2 - s) + \rho(m, e \mid I^{s},  k, \beta) (m-s).
 		\end{eqnarray*}
 	\end{small}
    An uniformed sender truthfully share his information rather than lie in favour of the receiver's favourite candidate whenever $s \geq \frac{1}{2} - \frac{m}{4} \frac{1-\sigma_{L}}{\sigma_{L} (1-x_{L})^{\beta k + 1} + 1 - \sigma_{L}}\equiv \overline{s}_l$, and the sender's preferences are aligned with his receiver's preferences $ s < 1/2$.
 	Again, the result of Lemma \ref{paircheaptalk} apply to the $k$-pairwise communication game in each matched pair. 
 	
 	\noindent From sequential rationality and consistency requirements that Lemma \ref{lemSeq} holds also for the k-player cheap talk game. 
 	
 	\noindent Then, if a voter receives at least one credible and informative message about the candidate's type of party $J$, that is $ M_{J}(t_J) = m $, he knows that the $J$'s candidate is a moderate with probability one. Conversely, if the receiver observes $k$ credible but uninformative messages about party $J$ candidate type, $ M_{J}(t_J)=\emptyset $, then he knows that the candidate of party $J$ is a moderate with probability $ \sigma_J (m) (1-x_{J}(m))^{k} / \sigma_J (m) (1-x_{J}(m))^{k} +  1 -\sigma_J (m) < 1$.
 	
 	\noindent Consider a voter whose bliss point is $i > 1/2$, then $i$ is an ex-post independent voter if and only if the following inequality holds
 	\begin{equation*}
 	u_i(v_i=1 \mid \tilde{I}, x, k) \geq u_i (v_i =0 \mid \tilde{I}, x, k), 
 	\end{equation*}
 	where $\tilde{I} = ( t_L, M_L(t_L), t_R, M_R(t_R) )$ is either $( m, M_L(t_L), \emptyset, \emptyset ) $ or $( t_L, m, \emptyset, \emptyset )$ for at least one credible message $M_L(t_L)$ and all $k$ credible messages $M_R(t_R)$, which holds whenever 
 	\begin{displaymath}
 	i \leq \frac{1}{2} + \frac{m}{4} \frac{1-\sigma_{R}}{1-\sigma_{R} + \sigma_{R} (1-x^{*}_{R}(m))^{\beta k + 1}} \equiv q_{r}(k,\beta, x_R).
 	\end{displaymath}
 	By symmetry, a voter biased toward the left party is an ex-post independent voter if
 	\begin{displaymath}
 	i\geq \frac{1}{2} - \frac{m}{4}\frac{1-\sigma_{L}}{1-\sigma_{L} + \sigma_{L} (1-x^{*}_{L}(m))^{\beta k + 1} } \equiv q_{l}(k,\beta, x_L).
 	\end{displaymath}
 	It is immediate to notice that the two cut-offs that define the ideology of ex-post independent voters $(q_{l}(k,\beta, x_L) ,q_{r}(k,\beta, x_R) )$ coincide with the cutoffs $\overline{s}_l$ and $\overline{s}_r$, where $\overline{s}_r = \frac{1}{2} + \frac{m}{4} \frac{1-\sigma_{R}}{1-\sigma_{R} + \sigma_{R} (1-x^{*}_{R}(m))^{\beta k + 1}}$ given the symmetry in the communication game.

 	\noindent 	This prove my result. Voters, whose ideological bliss point is between this two cut-off points truthfully and credibly transmit information whenever their preferences are aligned.
 	\hfill $\blacksquare$

\bigskip

    In what follows, to ease the proof of both Theorems \ref{Random} and \ref{target}, I assume party $L$ has a moderate candidate $t_{L} = m$.

 	\paragraph{Proof of Theorem \ref{Random}.} 
    Parties can only randomly advertise their candidate.   
    If both parties have a moderate candidate, in the symmetric equilibrium, each wins the election with probability $ 1/2 $. If party $R$ has an extremist $t_{R} = e$ what matter is the distribution of indifferent voters for each possible information set, $ I = \{ ( m, m , \emptyset, \emptyset), ( \emptyset, m , \emptyset, \emptyset),$ $  ( m, \emptyset , $  $\emptyset,$ $\emptyset),$  $ ( \emptyset,$  $ \emptyset, $ $ \emptyset, \emptyset)   \} $, hold by voters. 
 	The indifferent voter for each information is 
 	\begin{equation*}
 	i^{*}(I) =  \frac{1}{2} + \frac{1}{2} \displaystyle \left( \sum_{t_{L}, t_{R} \in T}   \rho ( t_{L}, t_{R} \mid \theta, I, x, k, \beta ) \, t_{L} - \rho ( t_{L},t_{R} \mid \theta, I, x,  k, \beta ) \, t_{R} \right)  , 
 	\end{equation*}
 	which implies
 	\begin{displaymath}
 	\pi(m,e \mid x) = \frac{4 +  \rho(m, e \mid I, x, k, \beta) (1 - (1-x_{L})^{\beta k + 1})}{8}.
 	\end{displaymath}
 	Party $L$'s expected utility from randomly advertising a moderate candidate is 
 	\begin{small}
 	\begin{eqnarray*}
 		U_{L}(x_{L} \mid \theta, k, \beta) & = &  \sigma_{R} \Big[\frac{1}{2}(1-2 m)- \frac{2-3m}{2} \Big] \\
 		& + & (1-\sigma_{R}) \Big[ 4 +  \frac{2-3m}{16} \rho(m, e \mid I, x, k, \beta)  (1 - (1-x_{L})^{\beta k + 1})  - (1-m) \Big] - c x_{L}
 	\end{eqnarray*}
 	\end{small} 
 	and the optimal level of advertisement  $ x^{*}_{L}(m) \in (0,1) $ solves
 	\begin{equation}\label{focxmm}
 	(1- \sigma_{R}(m)) (\beta k + 1) \rho (m, e \mid I, \beta, k) (1-x_{L})^{\beta k} = \frac{16 c}{2 -3 m}.
 	\end{equation}
 	if $(1 - \sigma_{R}) (\beta k + 1) \rho(m, e \mid \theta, I, x, k, \beta)  (2 -3 m) / 16 > c $ party $L$ will never randomly advertise a moderate candidate and $ x_{L} = 0 $, otherwise it will advertise a moderate if and only if the expected utility form advertising him  with intensity $x^{*}_{L}$, $U^{\emptyset}_{L}(x_{L}^{*} \mid \theta, k, \beta)$, is higher than the expected utility of not advertising him, $ U^{\emptyset}_{L}(\tilde{x} \mid \theta, k, \beta) $, where $ \tilde{x}_{L}(m) = 0 $.
 	The expected payoff of the left party when she does observe that its candidate is a moderate, and decides to not advertise is
 		\begin{small}
 	\begin{eqnarray*}
 		U^{\emptyset}_{L}(\tilde{x} \mid \theta, k, \beta) & = & \sigma_{R} \Big[ \Big(\frac{1}{2} - \frac{\rho(e,m \mid I, \beta, k) (1- (1-x_{R})^{\beta k + 1})}{8} \Big) (1- 2m) - \frac{2- 3m}{2} \Big] \\
 		& + &  (1-\sigma_{R}) \Big[ \frac{1}{2} \frac{2- 3m}{2} + (1- m) \Big].
 	\end{eqnarray*}
 	 	\end{small} 
 	Thus $ U_{L}(x_{L} \mid \theta, k, \beta) \geq U^{\emptyset}_{L}(\tilde{x} \mid \theta, k, \beta) $ if and only if 
 	\begin{equation}\label{disxm}
 	c x_{L} < \dfrac{ \rho(e , m \mid  I, \beta, k) (2 - 3 m + \sigma_{R} m) (1- (1-x_{L})^{(\beta k + 1)})}{16}.
 	\end{equation}
 	Thus, $x^{*}_{L}(m)$ solves
 	\begin{displaymath}
 	(1- \sigma_{R}(m)) (\beta k + 1) \rho(m, e \mid I, \beta, k) (1-x_{L})^{\beta k} = \frac{16 c}{2 -3 m},
 	\end{displaymath}
 	and party randomly advertises a moderate candidate whenever
 	\begin{displaymath}
 	c< c^{*}(k, \beta) \leq \frac{ (1- \sigma_{R}(m)) (\beta k + 1) \rho(m, e \mid I, \beta, k)  (2 -3 m)}{ 16}
 	\end{displaymath}
 	and never advertises a moderate candidate otherwise. 
\hfill $\blacksquare$

\bigskip

        \paragraph{Proof of Theorem \ref{target} } 
    Each party, when using a targeting technology, can target their advertisement only to the $L$-leaning or $R$-leaning voters. If a party targets its advertisement to $J$-leaning voters then all of them receive the advertisement with probability one while no $J^{-}$-leaning voter receives the advertisement. 
 	
 	\noindent Let $ x_{Jj}^d (t_{J} )=  1 $ denote the advertising strategy of a party $J$ who target ($d$) its advertisement to voters leaning toward party $j = \{ L , R\}  $.
    
    To simplify the argument assume that party $R$ has no access to the targeting technology and it randomly advertises a moderate candidate if $ c < c^{*}(k, \beta) $.

    First I show that parties never target their advertisement to their own supporters. Consider the case in which party $L$ targets its own supporters, then all $L$-leaning voters, both partisans and independents, are informed about the candidate's type of their favoured party and, regardless of the candidate type and advertising strategy of party $R$, party $L$ wins the election with probability $ 1/2 $ it would win if it does not advertise its candidate. 
    
    \noindent It follows that if  $ c > c^{*}(k, \beta) $, which implies party $R$ never advertises a moderate candidate, party $L$ prefers to do not advertise a moderate candidate either, as it will win the electoral competition with probability $1/2$ without paying any advertisement cost.
    If, instead, $ c < c^{*}(k, \beta) $, party $L$ by randomly advertising a moderate candidate wins the electoral competition with probability $ \pi_{L}(m, m \mid x) = 1/2 $ when it faces a moderate opponent, and with probability $ \pi_{L}(m,e \mid x) = 1/2 + \rho(m , e \mid I, \beta, k) (1- (1-x_{L})^{(\beta k + 1)} > 1/2 $ when it faces an extremist opponent. Then, party $L$ target the advertisement to its own supporter rather than randomly advertise its candidate if and only if $U_{L}( x_{LR}^d (m) \mid \theta, \beta , k) \geq U_{L}$ $(x_{L}^* \mid \theta, \beta, k) $, i.e. if and only if
 	\begin{eqnarray*}
 		(1 & - & \sigma_{R}) \Big[ \Big( \frac{  \rho (m,e \mid I, \beta, k) (1-(1-x_{R})^{\beta k + 1})}{8} -\frac{1}{2} \Big) \Big(\frac{2-3m}{2} \Big) \Big] \\
 		& + & \sigma_{R} \Big( \frac{\rho (m,e \mid I, \beta, k) (1- 2 (1-x_{R})^{\beta k + 1}}{8} \Big) (1-2m) > 0,
 	\end{eqnarray*}
 	which never holds for any $ x_{R} \in (0,1)$.
    Thus, in a symmetric pure strategy equilibrium parties never target their advertisements to their own supporters. 
    
    Now consider the case in which party $L$ targets the advertisement to voters leaning for the opponent. 
    
    \noindent If  $ c > c^{*}(k, \beta) $ parties never randomly advertise their candidate. 
    Suppose party $R$ does not advertise its candidate and only party $L$ has access to targeting technology. Party $L$ by targeting its advertisement to the opponent's supporter wins the elections with probability one. It is trivial to show that party $L$ deviate from $x_L(m)=0$ whenever the cost of the advertisement is such that $c^{*}(k, \beta) < c < \frac{1-2m}{2}$, which is the case only when the credible source of voters information are not too many, {\em i.e.} $\beta k + 1 < \frac{16 c}{ (2-3m)(1-\sigma)(1- \rho(m,m \mid x, I, k ,\beta)}$. By symmetry it follows that also party $R$ has an incentive deviate by targeting $L$-leaning voters. This would imply that each party would win the election with probability $1/2$, if both run with a moderate candidate, and with probability one if their opponent has an extremist. It is easy to show that if $c^{*}(k, \beta) < c < \frac{1-2m}{2}$ it does not exist any pure strategy symmetric equilibrium.\footnote{Parties will mix between the two strategy with probability $\zeta = \frac{1-m +2c}{1-2m}$.} 
    
    \noindent If, instead, $ c < c^{*}(k, \beta) $ parties randomly advertise their moderate candidate. Suppose again only party $L$ has access to targeting technology. By targeting the opponent's supporters party $L$ wins the election with probability $ \pi_{L}^{R}(m,e \mid x^{d}_{LR}, x_{R}^*) = 1  $ when the opponent runs with an extremist, and with probability $ \pi_{L}(m,m \mid x^{d}_{LR}, x_{R}^*) = 1/2 - ( \rho (e \mid I, \beta, k) (1-2(1- x_{R}))/16 < 1/2 = \pi_{L}(m,m \mid x_{L}^{*}, x_{R}^*) $ if the opponent runs with a moderate. 
    Then, party $L$ targets the advertisement to the $R$-leaning voters  if and only if $U_{L}( x_{LR}^d (m) \mid \theta, \beta, k) \geq U_{L}(x_{L} \mid \theta, \beta , k)$, {\em i.e.} if and only if
 	\begin{eqnarray*}
 		\sigma_{R} \Big[ \frac{1}{2} (1- 2m) - \frac{2-3m}{2} \Big] + (1- \sigma_{R}) \Big[ \frac{2-3m}{2} - (1-m) \Big] - c  \geq \\
 		\sigma_{R} \Big[ \frac{1}{2} (1-2m)- \frac{2-3m}{2} \Big] + (1- \sigma_{R}) \Big[ \frac{1}{2} \frac{2-3m}{2} - (1-m) \Big],
 	\end{eqnarray*}
 	where recall $x^{d}_{LR} = 1$, and which implies 
 	\begin{displaymath}
 	c < \hat{c} \equiv \frac{(2-3m)(1-\sigma_{R})}{ 4 }. 
 	\end{displaymath}

 	The party $R$ also has incentive to target $L$-leaning voters whenever $c < \hat{c}$. This implies that, as in the benchmark model (Section 3) the cost of advertisement must be 
 	\begin{equation}
 	    c < \hat{\bar{c}} \equiv \frac{(2-3m - \sigma m)}{ 4 }.
 	\end{equation}

    Looking at the cost thresholds it is easy to show that $c^*(\beta, k) \geq \hat{\bar{c}}$ whenever voters' personal network is characterized by high degree of homophily and/or the richness of the network, {\em i.e.} if {\small
    $$k \beta \geq \bar{k \beta} \equiv \frac{(2-3m)\left( (2+(1-
\rho(m,m \mid I, \beta, k))\sigma)+(1+\rho(m,m \mid I, \beta, k))\right) -4m \sigma)}{(2-3m)(1-\rho(m,m \mid I, \beta, k))(1- \sigma) }.$$ }
Whenever this is the case the network it acts as an information diffusion device and by randomly advertising their candidate parties can exploit the echo chambers.
    If, instead, voters' personal network is characterized by few sources of credible information $c^*(\beta, k) < \hat{c}$, and  parties have higher incentive in targeting the opponent's supporters as the voters' personal network entails a significant waste of information. 
 	\hfill $\blacksquare$

   \bibliographystyle{plainnat}
    
   \bibliography{echoct}

    \end{document}